\documentclass{article}

\usepackage{arxiv}

\usepackage[utf8]{inputenc} 
\usepackage[T1]{fontenc}    
\usepackage{hyperref}       
\usepackage{url}            
\usepackage{booktabs}       
\usepackage{amsmath,amsthm,amssymb}

\usepackage{algorithm2e}
\usepackage{comment}

\usepackage{nicefrac}       
\usepackage{microtype}      
\usepackage{lipsum}
\usepackage{graphicx}

\newcounter{theorem}
\newtheorem{thm}[theorem]{Theorem}

\newtheorem{lemma}[theorem]{Lemma}

\newtheorem{cor}[theorem]{Corollary}

\newtheorem*{definition}{Definition}

\newtheorem*{remark}{Remark}
 
{
      \theoremstyle{plain}
      \newtheorem*{assumption1a}{(A.1)}
      \newtheorem*{assumption2a}{(A.2)}
       \newtheorem*{assumption3a}{(A.3)}
      \newtheorem*{assumption4a}{(A.4)}
      
       \newtheorem*{assumptionH1}{(H.1)}
      \newtheorem*{assumptionH2}{(H.2)}
       \newtheorem*{assumptionK1}{(K.1)}
      \newtheorem*{assumptionK2}{(K.2)}
  }
\newcommand{\E}{\mathbb{E}}

\newcommand{\var}{\mathbb{V}\mathrm{ar}} 
\newcommand{\tr}{\mbox{tr}}

\newcommand{\Var}{\mathbb{V}\mathrm{ar}}
\newcommand{\cov}{\mbox{cov}}
\newcommand{\rank}{\mbox{rank}}
\newcommand{\spn}{\operatorname{span}}

\newcommand{\diag}{\mbox{diag}}
\newcommand{\spc}{{\mathcal S}}

\newcommand{\real}{{\mathbb R}}
\newcommand{\1}{\mathbf 1}

\newcommand{\X}{{\mathbf X}}
\newcommand{\I}{{\mathbf I}}
\newcommand{\U}{{\mathbf U}}
\newcommand{\Ob}{{\mathbf O}}

\newcommand{\x}{{\mathbf x}}
\newcommand{\bb}{{\mathbf b}}

\newcommand{\M}{{\mathbf M}}
\newcommand{\W}{{\bf W}}

\newcommand{\G}{{\mathbf G}}
\newcommand{\0}{{\bf 0}}
\newcommand{\B}{{\mathbf B}}
\newcommand{\T}{{\mathbf T}}

\newcommand{\V}{{\mathbf V}}

\newcommand{\eb}{\mathbf{e}}

\newcommand{\Pbf}{\mathbf{P}}

\newcommand{\R}{{\mathbf R}}

\newcommand{\vb}{\mathbf v}
\newcommand{\bs}{\mathbf s}

\newcommand{\rs}{{\mathbf r}}
\newcommand{\y}{{\mathbf y}}

\def\xn{{\mathbf x}}

\newcommand{\Pb}{\mathbb{ P}}

\newcommand{\greekbold}[1]{\mbox{\boldmath $#1$}}

\newcommand{\Sigmabf}{\greekbold{\Sigma}}
\newcommand{\Sigmaxbf}{\Sigmabf_{\x}}

\newcommand{\argmin}{\operatorname{argmin}}

\title{\bf Conditional Variance Estimator for Sufficient Dimension Reduction}

\author{Lukas Fertl\thanks{lukas.fertl@tuwien.ac.at}\\
            \hspace{.2cm}
			Institute of Statistics and Mathematical Methods in Economics\\ Faculty of Mathematics and Geoinformation\\ TU Wien, Vienna, Austria
			\And \\
			\textbf{Efstathia Bura} \thanks{efstathia.bura@tuwien.ac.at} \\
			\hspace{.2cm}
			Institute of Statistics and Mathematical Methods in Economics\\ Faculty of Mathematics and Geoinformation\\ TU Wien, Vienna, Austria}

\begin{document}
\maketitle
\begin{abstract}
\textit{Conditional Variance Estimation} (CVE) 
is a novel sufficient dimension reduction (SDR) 
method for additive error regressions with continuous predictors and link function. It operates under the assumption that the predictors can be replaced by a lower dimensional projection without loss of information. In contrast to the majority of moment based sufficient dimension reduction  methods, Conditional Variance Estimation  is fully data driven, does not require the restrictive linearity and constant variance conditions, and is not based on inverse regression. CVE is shown to be consistent and its objective function to be uniformly convergent. CVE  outperforms the mean average variance estimation, (MAVE), 
its main competitor, in several simulation settings, remains on par under others, while it always outperforms the usual inverse regression based linear SDR methods, such as Sliced Inverse Regression.
\end{abstract}


\section{Introduction}

Suppose $(Y,\X^T)^T$ have a joint continuous distribution, where $Y \in \real$ denotes a univariate response and $\X \in \real^p$ a $p$-dimensional covariate vector. We assume that the dependence of $Y$ and $\X$ is modelled by
\begin{align}
Y = g(\B^T \X) + \epsilon, \label{mod:basic}
\end{align}
where $\X$ is independent of $\epsilon$ with positive definite variance-covariance matrix, $\var(X)=\Sigmaxbf$, $\epsilon  \in \real$ is a mean zero random variable with  finite $\var(\epsilon)= E\left(\epsilon^2\right)=\eta^2$, $g$ is an unknown continuous non-constant function, and $\B= (\bb_1, ..., \bb_k) \in \real^{p \times k}$ of rank $k \leq p$.
Model \eqref{mod:basic} states that 
\begin{align}\label{meanSDR}
\E(Y\mid \X)&= \E(Y \mid \B^T \X)    
\end{align}
and requires the first conditional moment $\E(Y \mid \X)=g(\B^T \X)$ contain the entirety of the information in $X$ about $Y$ and be captured by $\B^T \X$, so that 
$F(Y\mid \X)=F(Y \mid \B^T \X)$, 
where $F(\cdot \mid \cdot)$ denotes the conditional cumulative distribution function (cdf) of  the first given the second argument. That is, $Y$ is statistically independent of $\X$ when $\B^T \X$  is given and replacing $\X$  by $\B^T \X$ induces no loss of information for the regression of $Y$ on $\X$. 

Identifying  the span of $\B$; i.e., the column space of $\B$, as only the $\spn\{ \B \}$ is identifiable, 
suffices in order to identify the \textit{sufficient reduction} of $\X$ for the regression of $Y$ on $\X$.  
We assume, without loss of generality, $\B$ is semi-orthogonal, i.e., $\B^T \B = \I_k$, since a change of coordinate system by an orthogonal transformation does not alter  model~\eqref{meanSDR}. 

For $q \leq p$, let 
\begin{equation}\label{Smanifold}
    \spc(p,q) = \{\V \in \real^{p \times q}: \V^T\V = \I_q\},
\end{equation}
denote the Stiefel manifold, that comprizes of all $p \times q$ matrices with orthonormal columns.  $\spc(p,q)$ is compact and  $\dim(\spc(p,q)) = pq - q(q+1)/2$  [see  \cite{Boothby} and Section 2.1 of \cite{Tagare2011}]. Further let
\begin{equation}\label{Grassman_def}
    Gr(p,q) = \spc(p,q)/\spc(q,q)
\end{equation}
denote the Grassmann manifold, i.e. all $q$-dimensional subspaces in $\real^p$, which is exactly the quotient space of $\spc(p,q)$ with all $q \times q$ orthonormal matrices $\spc(q,q)$, i.e. the basis of a linear subspace is unique up to orthogonal transformations.

The fact that only $\spn\{\B\}$ is identifiable, can be expressed through the Grassmann manifold $Gr(p,q)$ in \eqref{Grassman_def}. The goal of sufficient dimension reduction in model \eqref{mod:basic} is to find a subspace $\M \in Gr(p,k)$ such that any basis $\B \in \spc(p,k)$ of $\M$ fulfills \eqref{mod:basic} or equivalently \eqref{meanSDR}.

Finding sufficient reductions of the predictors to replace them in regression and classification without loss of information is called \textit{sufficient dimension reduction} 
\cite{Cook1998}. The first split in sufficient dimension reduction taxonomy occurs between likelihood and non-likelihood based methods. The former, which were developed more recently  \cite{Cook2007, CookForzani2008, CookForzani2009, BuraForzani2015, BuraDuarteForzani2016}, assume knowledge either of the joint family of distributions of $(Y,\X^T)^T$, or the conditional family of distributions for $\X \mid Y$. The latter is the most researched branch of sufficient dimension reduction and comprizes of three classes of methods: Inverse regression based, semi-parametric and nonparametric. Reviews of the former two classes can be found in \cite{AdragniCook2009,MaZhu2013, Li2018}. 

In this paper we present the \textit{conditional variance estimation}, which  falls in the  class of nonparametric methods. The estimators in this class 
minimize a criterion that describes the fit of the dimension reduction model \eqref{meanSDR} under  \eqref{mod:basic} to the observed data. Since the criterion involves unknown distributions or regression functions, nonparametric estimation is used to recover  $\spn\{\B\}$. 
Statistical approaches to identify $\B$ in \eqref{meanSDR} include ordinary least squares and nonparametric multiple index models \cite{multiIndexModel}.  The least squares estimator, $\Sigmaxbf^{-1} \cov(\X,Y)$, always falls in $\spn\{\B\}$ \cite[Th. 8.3]{Li2018}. Principal Hessian Directions \cite{Li1992} was the first sufficient dimension reduction estimator to target $\spn\{\B\}$ in \eqref{meanSDR}. Its main disadvantage is that it requires the so called \textit{linearity} and \textit{constant variance} conditions on the marginal distribution of $\X$. Its relaxation, Iterative Hessian Transformation \cite{CookLi2004}, still requires the linearity condition in order to recover vectors in $\spn\{\B\}$.  

The most competitive nonparametric sufficient dimension reduction method up to now has been  \textit{minimum average variance estimation} (MAVE, \cite{Xiaetal2002}). It assumes  model \eqref{mod:basic}, 
bounded fourth derivative covariate density, and existence of continuous bounded third derivatives for $g$. It uses a local first order approximation of $g$ in \eqref{mod:basic} and  minimizes the expected conditional variance of the response given $\B^T \X$.


The \textit{conditional variance estimator}  also targets and recovers $\spn\{B\}$ in models \eqref{mod:basic} and \eqref{meanSDR}.  The objective function is based on the intuition that the directions in the predictor space that capture the dependence of $Y$ on $X$ should exhibit significantly higher variation in $Y$ as compared with the directions along which $Y$ exhibits markedly less variation. The \textit{conditional variance estimator} is a fully data-driven estimator that performs better than or is on par with \textit{minimum average variance estimation} in simulations. 
The \textit{conditional variance estimator} differs from other approaches, including MAVE, in that it only targets the $\spn\{\B\}$ and does not require an explicit form or estimation of the link function $g$. As a result, it requires weaker assumptions on its smoothness. 

\section{Motivation}\label{motivation}

Let  $(\Omega ,{\mathcal {F}},P)$ be a probability space, and $ \X:\Omega \rightarrow \real^p$  be a random vector with a continuous probability density function $f_{\X}$ and denote its support by $\mbox{supp}(f_{\X})$. Throughout $\|\cdot\|$ denotes the Frobenius norm for matrices, Euclidean norm for vectors, and scalar product refers to the euclidean scalar product. For any matrix $\M$, or linear subspace $\M$, we denote by $\Pbf_{\M}$ the projection matrix on the column space of the matrix or on the subspace, i.e. $\Pbf_{\M} = \M(\M^T \M)^{-1} \M^T \in \real^{p \times p}$ for $\M \in \real^{p \times q}$. For any $\V \in \spc(p,q)$, defined in \eqref{Smanifold}, we generically denote a basis of the orthogonal complement of its column space $\spn\{\V\}$, by $\U$. That is, $\U \in \spc(p,p-q)$ such that $\spn\{\V\} \perp \spn\{U\}$ and $\spn\{\V\} \cup \spn\{\U\} = \real^p$, $\U^T\V = \0 \in \real^{(p-q) \times q}, \U^T\U = \I_{p-q}$. For any $\xn, \bs_0 \in \real^p$ we can always write
\begin{equation}\label{ortho_decomp}
    \xn = \bs_0 + \Pbf_\V (\xn - \bs_0) + \Pbf_U (\xn - \bs_0) = \bs_0 + \V\rs_1 + \U\rs_2
\end{equation}
where $\rs_1 = \V^T(\xn-\bs_0) \in \real^{q}, \rs_2 = \U^T(\xn-\bs_0) \in \real^{p-q}$.

In the sequel, we refer to the following assumptions as needed and the proofs of the Theorems are presented in the Appendix.
\begin{assumption1a}
Model $Y = g(\B^T\X) + \epsilon$ 
holds with $Y \in \real$, $g:\real^k \to \real$ non constant in all arguments, $\B= (\bb_1, ..., \bb_k) \in \real^{p \times k}$ of rank $k \leq p$, $\X \in \real^p$ independent from $\epsilon$, $ \var(\X) = \Sigmaxbf $ is positive definite , $\E(\epsilon)=0$, $\var(\epsilon)= \eta^2 < \infty$.
\end{assumption1a}
\begin{assumption2a}
The link function $g$ and the density $f_\X : \real^p \to [0,\infty)$ of $\X$ are twice continuous differentiable.
\end{assumption2a}
\begin{assumption3a}\label{A3}
$\E(|Y|^8) < \infty$.
\end{assumption3a}
\begin{assumption4a}\label{A4}
$\text{supp}(f_\X)$ is compact.
\end{assumption4a}
\begin{remark}
Assumption (A.4) is not as restrictive as it might seem.  \cite{CompactAssumption} showed in Proposition 11 that there is a compact set $\mathcal{S} \subset \real^p$ such that the mean subspace of model \eqref{mod:basic} is the same as the mean subspace of $Y = g(\B^T\X_{|\mathcal{S}}) + \epsilon$,  where $\X_{|\mathcal{S}} = \X 1_{\{\X \in \mathcal{S}\}}$ and $1_A$ is the indicator function of $A$. Further $\mathcal{S}$ can be assumed to be an ellipsoid and for all $\widetilde{\mathcal{S}} \supseteq \mathcal{S}$ the same assertion holds true.
\end{remark}


\begin{definition}
For $q \leq p \in N$ and any $\V \in \spc(p,q)$, we define
\begin{equation}\label{Lvs}
\tilde{L}(\V, \bs_0) = \Var(Y\mid\X \in \bs_0 + \spn\{\V\}),
\end{equation}
where $\bs_0 \in \real^p$ is a shifting point.
\end{definition}
 
\begin{definition}
For $\V \in \spc(p,q)$, we define the objective function,
\begin{equation}\label{objective}
L(\V)= \int_{\real^p}\tilde{L}(\V,\xn)f_{\X}(\xn)d\xn = \E\left(\tilde{L}(\V,\X)\right).
\end{equation} 
\end{definition}
$L(\V)$ in \eqref{objective} is the objective function for the estimator we propose for the span of $\B$ in \eqref{mod:basic} and Theorem~\ref{thm1} provides the statistical motivation for the objective function \eqref{objective} of the conditional variance estimator. First we derive that both population based functions \eqref{Lvs} and \eqref{objective} are well defined. 

Let $\X$ be a $p$-dimensional  continuous random vector with density $f_\X(\xn)$, $\bs_0 \in \text{supp}(f_\X) \subset \real^p$, and $\V$ belongs to the Stiefel manifold $\spc(p,q)$ defined in \eqref{Smanifold}. 
The function \begin{gather}\label{density}
f_{\X\mid\X \in \bs_0 +\spn\{\V\}}(\rs_1) =
\frac{f_\X(\bs_0 + \V\rs_1)}{\int_{\real^q}f_\X(\bs_0 + \V\rs)d\rs} 
\end{gather}
is a proper conditional density of $\X$ 
that is concentrated on the affine subspace $\bs_0 + \spn\{\V\}$ using the concept of regular conditional probability \cite{Leaoetal2004} under assumption (A.2). The detailed justification is given in the Appendix, where we also show that under  assumptions (A.1), (A.2) and (A.4),  $\tilde{L}(\V, \bs_0)$ in \eqref{Lvs} and $L(\V)$ in \eqref{objective} are well defined and continuous. Moreover, 
\begin{align}\label{LtildeVs0}
\tilde{L}(\V,\bs_0) = \mu_2(\V,\bs_0) - \mu_1(\V,\bs_0)^2 + \eta^2
\end{align} 
where 
\begin{align}\label{mu_l}
\mu_l(\V,\bs_0) &= \int_{\real^q} g(\B^T\bs_0 + \B^T\V\rs_1)^l\frac{f_\X(\bs_0 + \V\rs_1)}{\int_{\real^q}f_\X(\bs_0 + \V\rs)d\rs}d\rs_1 = \frac{t^{(l)}(\V,\bs_0)}{t^{(0)}(\V,\bs_0)}
\end{align}
with \begin{align}\label{tl}
t^{(l)}(\V,\bs_0) &= \int_{\real^q} g(\B^T\bs_0 + \B^T\V\rs_1)^l f_\X(\bs_0 + \V\rs_1)d\rs_1.
\end{align}

\begin{thm}\label{thm1}
Suppose $\V = (\vb_1,...,\vb_q) \in \spc(p,q)$ 
and $q\in \{1,\ldots,p\}$. Under assumptions (A.1), (A.2) and (A.4),
\begin{itemize}
\item[(a)] For all $\bs_0 \in \real^p $ and $\V $ such that there exist $u \in \{1,...,q\}$ with  $\vb_u \in \spn\{\B\}$, $\tilde{L}(\V,\bs_0) > \Var(\epsilon)= \eta^2$ and $L(\V) >  \eta^2$.
\item[(b)]  For all $\bs_0 \in \real^p$ and $\spn\{\V\} \perp \spn\{\B\}$, $\tilde{L}(\V,\bs_0) = \eta^2$ and $L(\V) =  \eta^2$.

\end{itemize}
\end{thm}

\begin{proof}

Let $\bs_0 \in \real^p$ and $\V = (\vb_1,...,\vb_q) \in \real^{p \times q}$  so that $\vb_u \in \spn\{\B\}$ for some $u \in \{1,...,q\}$. 
To obtain (a), observe $\X \in \bs_0 +\spn\{\V\} \Longleftrightarrow \X = \bs_0 + \Pbf_{V}(\X - \bs_0)$ and using \eqref{Lvs} yields 
\begin{align}
\tilde{L}(\V,\bs_0) 
&= \Var\left(g(\B^T\X)\mid\X = \bs_0 + \V\V^T(\X-\bs_0)\right) + \Var(\epsilon) \notag \\
&= \Var\left(g(\B^T\bs_0 + \B^T\V\V^T(\X-\bs_0))\mid\X = \bs_0 + \V\V^T(\X-\bs_0)\right) + \eta^2 >\eta^2 \label{1stterm}
\end{align}
since $\B^T\V\V^T(\X-\bs_0) \neq 0$ with probability 1, and therefore the variance term in \eqref{1stterm} is positive.  For  $\V$ such that  $\V $ and $\B$ are orthogonal,  $\B^T\V\V^T(\X-\bs_0) = 0$ and (b) follows. Since $\bs_0$ is arbitrary yet constant, the statements for $L(\V)$ follow.
\end{proof}

 Theorem~\ref{thm1} also has an intuitive geometrical interpretation for the proposed method. If $\X$ is not random, the deterministic function $Y = g(\B^T\X)$ is constant in all directions orthogonal to $\B$ and varies in all other directions. If randomness is introduced,  as in model  \eqref{mod:basic}, 
then the variation in $Y$ stems only from $\epsilon$ in all directions orthogonal to $\B$. In all other directions the variation comprizes of the sum of the variation of $\epsilon$ and of $g(\B^T\X)$. 
In consequence, the objective function \eqref{objective} captures the variation of $Y$ as $\X$ varies in the column space of $\V$ and is minimized in the directions orthogonal to $\B$.

\subsection{Conditional Variance Estimator (CVE)}
We have shown that the objective function $L(\V)$ in \eqref{objective} is well defined and  continuous in Section \ref{motivation}.  Let 
\begin{equation}\label{optim}
\V_q = \argmin_{\V \in \spc(p,q)}L(\V).
\end{equation}
$\V_q$ is well defined as the minimizer of a continuous function over the compact set $\spc(p,q)$. Nevertheless, $\V_{q}$ is not unique since for all orthogonal $\Ob \in \real^{q \times q}$ such that  $\Ob\Ob^T = \I_{q}$, $L(\V \Ob) = L(\V)$ as $L(\V)$ depends on $\V$ only through $\spn\{\V\}$. Nevertheless, it is a unique minimizer over the Grassmann manifold $Gr(p,q)$ in \eqref{Grassman_def}. To see this, suppose $\V \in \spc(p,q)$ is an arbitrary basis of a subspace $\M \in Gr(p,q)$. We can identify $M$ through the projection $\Pbf_\M = \V\V^T$. 
By \eqref{ortho_decomp} we write $\xn = \V \rs_1 + \U \rs_2$. By the Fubini-Tornelli Theorem we obtain \begin{align}\label{Grassman}
    \tilde{t}^{(l)}(\Pbf_\M,\bs_0) &= \int_{\text{supp}(f_\X)} g(\B^T\bs_0 + \B^T P_\M \xn)^l f_\X(\bs_0 + \Pbf_\M \xn)d\xn \\&= t^{(l)}(\V,\bs_0) \int_{\text{supp}(f_\X)\cap \real^{p-q} }d\rs_2. \notag
\end{align}
Therefore $\tilde{t}^{(l)}(\Pbf_\M,\bs_0)/\tilde{t}^{(0)}(\Pbf_\M,\bs_0) = t^{(l)}(\V,\bs_0)/t^{(0)}(\V,\bs_0)$ and $\mu_l(\cdot,\bs_0)$ in \eqref{mu_l} can also be viewed as a function from $Gr(p,q)$ to $\real$. If the optimization~\eqref{optim} is over $Gr(p,q)$,  the objective function \eqref{objective} has a unique minimum at $\spn\{\B\}^\perp$ by Theorem~\ref{thm1}. Therefore $\B$ is not uniquely identifiable but its $\spn\{\B\}$ is. 

Corollary~\ref{cor1} follows directly from Theorem~\ref{thm1} and provides the means for identifying the linear projections of the predictors satisfying \eqref{mod:basic}.

\begin{cor}\label{cor1}
Under the assumptions (A.1), (A.2), and (A.3) the solution of the optimisation problem $\V_q$ in  \eqref{optim} is well defined. Let $k=\dim(\spn\{\B\})$ and $q=p-k$, 
\begin{enumerate}
\item[(a)] $\spn\{\V_{q}\} = \spn\{\B\}^\perp$
\item[(b)] $\spn\{\V_{q}\}^\perp = \spn\{\B\}$ \label{est_equ}
\end{enumerate}
\end{cor}

We next define the novel estimator of the sufficient reduction space, $\spn\{\B\}$, in \eqref{mod:basic}, which is motivated by Theorem~\ref{thm1} and
Corollary~\ref{cor1} (b) serves as the estimation equation for the conditional variance estimator at the population level. 

\begin{definition}
The \textbf{Conditional Variance Estimator}  is defined to be any basis  ${\B}_{p-q}$ of $\spn\{{\V}_{q}\}^\perp$.  That is, the CVE of $\B$ is any $\B_{p-q}$ such that
\begin{equation}\label{CVE}
\spn\{{\B}_{p-q}\} = \spn\{{\V}_{q}\}^\perp
\end{equation}
\end{definition}
When $q=p-k$, where $k=\rank(\B)$ in \eqref{mod:basic}, then the CVE obtains the population $\spn\{\B\}$.
Alternatively, we can also target $\B$ directly by maximizing the objective function $L(\V)$. The downside of this approach is that $\X$ either needs to be standardized, or the conditioning argument needs to be changed to 
$\X = \bs_0 + \Pbf_{\Sigmaxbf^{-1}(\spn\{\V\})}(\X-\bs_0)$, where $\Pbf_{\M(\spn\{\V\})}$ is the orthogonal projection operator with respect to the inner product $\langle \xn,\y \rangle_\M = \xn^T\M \y$. 
In either case, the inversion of $\Sigmaxbf$ is required. Our choice of targeting the orthogonal complement avoids the inversion of $\Sigmaxbf$, and  the estimation algorithm in Section~\ref{Optim} can be applied to regressions with $p > n$ or $p \approx n$, where $n$ denotes the sample size. Additionally, targeting the complement has computational advantages. The dimension of the search space $\spn\{{\V}_{q}\}^\perp$ is $p-q$, which is smaller than the dimension of the direct target space in \eqref{CVE} when $q=p-k$ for small $k$, which is the appropriate setting in a dimension reduction context.   
\section{Estimation}\label{estimation}
Assume $(Y_i,\X_i^T)_{i=1,...,n}^T$ is an independent identical distributed sample  from model \eqref{mod:basic}. For $\V \in \spc(p,q)$ and $\bs_0 \in \real^p$, we define
\begin{align}
	d_i(\V,\bs_0)&= \|\X_i - \Pbf_{\bs_0 + \spn\{\V\}}\X_i\|^2 = \|\X_i -\bs_0\|^2 - \langle \X_i - \bs_0,\V\V^T(\X_i - \bs_0)\rangle  \notag\\
	&= \| (\I_p - \V\V^T)(\X_i - \bs_0)\|^2 = \| \Pbf_{\U}(\X_i - \bs_0)\|^2 \label{distance}
\end{align}
where $\langle \cdot, \cdot\rangle$ is the usual inner product in $\real^p$, $\Pbf_{\V}=\V\V^T$ and $\Pbf_{U}=\I_p-\Pbf_{\V}$ using the orthogonal decomposition given by \eqref{ortho_decomp}. 

Let $h_n \in \real_+$ be a sequence of bandwidths and we call the set $\spc_{\bs_0,\V}=\{\xn \in \real^p: \|\xn - \Pbf_{\bs_0 + \spn\{\V\}}\xn\|^2 \leq h_n\}$ a \textit{slice} that depends on both the shifting point $\bs_0$ and the matrix $\V$.
$h_n$ represent the squared width of a slice around the subspace $\bs_0 + \spn\{\V\}$ and fulfills the following assumptions. 
\begin{assumptionH1}
For $n \to \infty$, $h_n \to 0$
\end{assumptionH1}
\begin{assumptionH2}
For $n \to \infty$, $nh^{(p-q)/2}_n \to \infty$
\end{assumptionH2}
\begin{remark}
For obtaining the consistency of the proposed estimator (H.2) will be strengthened to $\log(n)/nh^{(p-q)/2}_n \to 0$.
\end{remark}
Let $K$ be a  function satisfying the following assumptions.
\begin{assumptionK1}
$K:[0,\infty) \rightarrow [0,\infty)$ 
is a non increasing and continuous 
function, so that $|K(z)| \leq M_1$, with $\int_{\real^{q}} K(\|\rs\|^2) d\rs < \infty$ for $q \leq p-1$.
\end{assumptionK1}
\begin{assumptionK2}
There exist positive finite constants $L_1$ and $L_2$ such that the kernel $K$ satisfies one of the following:
\begin{itemize}
    \item[(1)] $K(u) = 0$ for $|u| > L_2$ and for all $u, \tilde{u}$ it holds $|K(u) - K(\tilde{u})| \leq L_1 |u - \tilde{u}|$
    \item[(2)] $K(u)$ is differentiable with $|\partial_u K(u)| \leq L_1$ and for some $\nu > 1$ it holds $|\partial_u K(u)| \leq L_1 |u|^{-\nu}$ for $|u| > L_2$
\end{itemize}
\end{assumptionK2}

Examples of  functions that satisfy (K.1) and (K.2) include the Gaussian, $K(z) = c\exp(-z^2/2)$, the exponential, $K(z) = c\exp(-z)$, and the squared Epanechnikov kernel, $K(z) = c \max\{(1-z^2),0\}^2$ (i.e. polynomial kernels), where $c$ is a constant. The rectangular, $K(z) = c I(z\leq 1)$, does not fulfill the assumptions but will be mentioned for intuitive explanations. A list of further kernel functions is given in  \cite[Table 1]{Parzen1961}. 

\subsection{\texorpdfstring{The estimator of $L(\V)$ and its uniform convergence}{The estimator of L(V) and its uniform convergence}}\label{LV.est}

\begin{definition}
For $i=1,\ldots,n$, we define
\begin{equation}
w_i(\V,\bs_0) = \frac{K\left(\frac{d_i(\V,\bs_0)}{h_n}\right)}{\sum_{j=1}^nK\left(\frac{d_j(\V,\bs_0)}{h_n}\right)} \label{weights}
\end{equation}
\end{definition}

\begin{definition}
The sample based estimate of $\tilde{L}(\V,\bs_0)$ is defined as
\begin{equation}
\tilde{L}_n(\V,s_0) = \sum_{i=1}^n w_i(\V,\bs_0)(Y_i - \bar{y}_1(\V,\bs_0))^2 = \bar{y}_2(\V,\bs_0) - \bar{y}_1(\V,\bs_0)^2 \label{Ltilde}
\end{equation}
where $\bar{y}_l(\V,\bs_0) = \sum_{i=1}^n w_i(\V,\bs_0)Y^l_i$, $l=1,2$.
\end{definition}

\begin{definition}
The estimate of the objective function $L(\V)$ in \eqref{objective} is defined as 
\begin{equation}
L_n(\V) = \frac{1}{n} \sum_{i=1}^n \tilde{L}_n(\V,\X_i), \label{LN}
\end{equation}
where each data point $\X_i$ is a shifting point.
\end{definition}
To obtain insight as to the choice of $\tilde{L}_n(\V,\bs_0)$ in \eqref{Ltilde}, let us consider the rectangular kernel, $K(z) = 1_{\{z \leq 1\}}$. In this case, $\tilde{L}_n(\V,\bs_0)$  computes the empirical variance of the  $Y_i$'s  corresponding to the $\X_i$'s that are no further than $\sqrt{h_n}$ away from the affine space $\bs_0 + \spn\{\V\}$, i.e., 
$d_i(\V, \bs_0) = \|\X_i - \Pbf_{\bs_0 + \spn\{\V\}}\X_i\|^2 \leq h_n$.  If  a smooth kernel is used, such as the Gaussian 
in our simulation studies, then $\tilde{L}_n(\V,\bs_0)$ is also smooth, which allows the computation of gradients required to solve the optimization problem.

In Theorem \ref{thm_L_uniform} we state the conditions under which  $L_n(\V)$ in \eqref{LN} converges uniformly to its population counterpart in \eqref{objective}. This result will lead to the consistency of our estimator. 

\begin{thm}\label{thm_L_uniform}
Let $\tilde{a}^2_n = \log(n)/n$. Under 
(A.1), (A.2), (A.3), (A.4), (K.1), (K.2), (H.1), $a_n^2 = \log(n)/nh_n^{(p-q)/2} = o(1)$ , and $a_n/h_n^{(p-q)/2} = O(1)$, 
\begin{equation} \label{thm5_eq}
    \sup_{\V \in \spc(p,q)}\left|L_n(\V) - L(\V)\right| \to 0 \quad \text{in probability as}\,\, n \to \infty
\end{equation}
\end{thm}
\subsection{The Conditional Variance Estimator}\label{CVE.est}
Next we define the estimator we propose for $\spn\{\B\}$ in \eqref{mod:basic}. Our main theoretical result follows in Theorem~\ref{thm_consistency} which establishes the consistency of our estimator.

\begin{definition}\label{Vhat}
The sample based {\bf Conditional Variance Estimator} $\widehat{B}_{p-q}$ is any basis of $\spn\{\widehat{\V}_q\}^\perp$
where $\widehat{\V}_q = \argmin_{\V \in \spc(p,q)}L_n(\V).$
\end{definition}

\begin{thm}\label{thm_consistency}
Under  
(A.1), (A.2), (A.3), (A.4), (K.1), (K.2), (H.1), $a_n^2 = \log(n)/nh_n^{(p-q)/2} = o(1)$, and $a_n/h_n^{(p-q)/2} = O(1)$, 
$\spn\{\widehat{\B}_{k}\}$ is a consistent estimator for $\spn\{\B\}$ in model \eqref{mod:basic}; i.e.,
\begin{equation*}
\|\Pbf_{\widehat{\B}_k} - \Pbf_{\B}\| \to 0 \quad \text{in probability as } n \to \infty .
\end{equation*}

\end{thm}

\subsection{\texorpdfstring{Weighted estimation of $L(\V)$}{Weighted estimation of L(V)}}\label{weight_section}

The set of points $\{\xn \in \real^p: \|\xn - \Pbf_{\bs_0 + \spn\{\V\}}\xn\|^2 \leq h_n\}$ represents a \textit{slice} in the a subspace of $\real^p$ about $\bs_0+ \spn\{\V\}$. 
In the estimation of $L(\V)$ two different weighting schemes are used:
\begin{itemize}
    \item[(a)] 
    \textit{Within a slice}. The weights are defined in \eqref{weights} and are used to calculate \eqref{Ltilde}.
    \item[(b)] 
    \textit{Between slices}. Equal weights $1/n$ are used to calculate \eqref{LN}.
\end{itemize}
The choice of weights can be potentially influential. Especially the between weighting scheme can further be refined by assigning more weight to slices with  more points. This can be realized by altering \eqref{LN} to
\begin{align}
L^{(w)}_n(\V) &=  \sum_{i=1}^n \tilde{w}(\V,\X_i) \tilde{L}_n(\V,\X_i), \quad \mbox{with} \label{wLN}\\
\tilde{w}(\V,\X_i) &= \frac{\sum_{j=1}^n K(d_j(\V,\X_i)/h_n) - 1}{\sum_{l,u=1}^nK(d_l(\V,\X_u)/h_n) -n} = \frac{\sum_{j=1,j\neq i}^n K(d_j(\V,\X_i)/h_n) }{\sum_{l,u=1, l\neq u}^nK(d_l(\V,\X_u)/h_n)}\label{wtilde}
\end{align}

For example, if a rectangular kernel is used, $\sum_{j=1,j\neq i}^n K(d_j(\V,\X_i)/h_n)$ is the number of $\X_j$ ($j \neq i$) points in the slice corresponding to $\tilde{L}_n(\V,\X_i)$. Therefore this slice gets higher weight, if the number of $\X_j$ points in this slice is larger. That is, the more observations we use for estimating $L(\V,\X_i)$ the better its accuracy. The denominator in \eqref{wtilde} 
guarantees the weights $\tilde{w}(\V,\X_i)$ sum up to one.

\subsection{Bandwidth selection}	
The performance of conditional variance estimation  depends crucially on the choice of the bandwidth sequence $h_n$ that controls the bias-variance trade-off if the mean squared error is used as measure for accuracy, in the sense that the smaller $h_n$ is, the lower the bias and the higher the variance and vice versa. Furthermore, the choice of $h_n$ depends on $p$, $q$, the sample size $n$, and the distribution  of $\X$. We assume throughout the bandwidth satisfies assumptions (H.1) and (H.2). 
We will use Lemma~\ref{thmsigma} to derive a data-driven bandwidth we use in the computation of our estimator.

\begin{lemma}\label{thmsigma}
Let $\M$ be a $p \times p$ positive definite matrix. Then,
\begin{equation} 
\frac{\tr(\M)}{p}=\argmin_{s>0}\|\M - s \I_p\|\label{trace}
\end{equation}
\end{lemma}

\begin{proof}
Let $\U$ be the $p \times p$ matrix whose columns are the eigenvectors of $\M$ corresponding to its eigenvalues $\lambda_1\ge \ldots \ge \lambda_p>0$.  Then,  $\M = \U \diag(\lambda_1,...,\lambda_p) \U^T$, which implies $\|\M - s \I_p\|^2_2 = \|\diag(\lambda_1,...,\lambda_p) -s \I_p\|^2 = \sum_{l=1}^p (\lambda_l -s)^2$. Taking the derivative with respect to $s$, setting it to 0 and solving for $s$ obtains \eqref{trace}, since $\sum_{l=1}^p \lambda_l = \tr(\M)$. 
\end{proof}

If the predictors are multivariate normal, 
their joint  density is approximated by $N(\mu_\X, \sigma^2 \I_p)$ by Lemma~\ref{thmsigma}, with $\sigma^2 = \tr(\Sigmaxbf)/p$. This results in no  bandwidth dependence on $\V$ and leads to a rule for bandwidth selection, as follows.

Under $\X \sim N_p(\mu_\X,\sigma^2 \I_p)$,  $\widetilde{\X}_i = \X_i - \X_j \sim N_p(0, 2\sigma^2 \I_p)$ for $i \neq j$, where we suppress the dependence on $j$ for notational convenience. Since  all data are used as  shifting points, $d_i(\V,\X_j) = \|\X_i-\X_j\|^2 - (\X_i-\X_j)^T\V\V^T(\X_i-\X_j) = \|\widetilde{\X}_i\|^2 - \widetilde{\X}_i^T\V\V^T\widetilde{\X}_i$. Let
\begin{align}
\text{nObs} &= \E\left(\#\{i\in\{1,...,n\}: \widetilde{\X}_i \in \spn_{h}\{\V\}\}\right) \notag\\
&= 1 + (n-1)\Pb(d_1(\V,\X_2) \leq h) = 1 + (n-1)\Pb(\|\widetilde{\X}\|^2 - \widetilde{\X}^T\V\V^T\widetilde{\X} \leq h) \label{nobs}
\end{align}
where $\spn_{h}\{\V\} = \{\xn \in \real^p: \|\xn - \Pbf_{\spn\{\V\}}\xn\|^2\leq h\}$ and $\widetilde{\X} = \X - \X^*$, with $\X^*$ an independent copy of $\X$.
nObs is  the expected number of points in a slice. 
Given a user specified value for \text{nObs},  $h$ is the solution to \eqref{nobs}.

Let $\xn \in \real^p$. For any $\V \in \spc(p,q)$ in \eqref{Smanifold}, there exists an orthonormal basis $\U \in \real^{p \times (p-q)}$ of $\spn\{\V\}^\perp$ such that 
$\xn = \V\rs_1 + \U\rs_2$,  
by \eqref{ortho_decomp}.
Then, $\widetilde{\X} = \V\R_1 + \U\R_2$, with $\R_1 = \V^T\widetilde{\X} \sim N(0,2\sigma^2\I_q), \R_2 = \U^T\widetilde{\X} \sim N(0,2\sigma^2\I_{p-q})$, and $\widetilde{\X}^T\V\V^T\widetilde{\X} = \|\R_1\|^2$ and $\|\widetilde{\X}\|^2 = \|\R_1\|^2 + \|\R_2\|^2$. Therefore,
\begin{align}\label{chi}
\Pb\left(\|\widetilde{\X}\|^2 - \widetilde{\X}^T\V\V^T\widetilde{\X} \leq h\right) = \Pb(\|\R_2\|^2 \leq h) = \chi_{p-q}\left(\frac{h}{2\sigma^2}\right),   
\end{align}
where $\chi_{p-q}$ is the cumulative distribution function of a chi-squared random variable with $p-q$ degrees of freedom. Plugging \eqref{chi} in \eqref{nobs} obtains
\begin{align}\label{nobs2}
\text{nObs} = 1 + (n-1)\chi_{p-q}\left(\frac{h}{2\sigma^2}\right).
\end{align}
Solving \eqref{nobs2} for $h$ and Lemma~\ref{thmsigma} yield 
\begin{equation}\label{hn}
h_n(\text{nObs}) =  \chi_{p-q}^{-1}\left(\frac{\text{nObs}-1}{n-1}\right) \frac{2\tr(\widehat{\Sigma}_{\xn})}{p},
\end{equation}
where $\widehat{\Sigma}_{\xn}=\sum_{i} (\X_i -\bar{\X}) (\X_i -\bar{\X})^T/n$ and $\bar{\X} = \sum_i \X_i/n$.

In order to ascertain $h_n$ satisfies (H.1) and (H.2), a reasonable choice  is to set $\text{nObs} = \gamma(n)$ for a function $\gamma(\cdot)$ with $\gamma(n) \to \infty$, ${\gamma(n)}/{n} \leq 1$ and ${\gamma(n)}/{n} \to 0$. 
For example, $\text{nObs} = \gamma(n) = n^\beta$ with $\beta \in (0,1)$ can be used.

Alternatively, a plug-in bandwidth based on rule-of-thumb rules of the form $c s n^{-1/(4+k)}$, where $s$ is an estimate of scale and $c$ a number close to 1, such as Silverman's ($c=1.06$, $s=$standard deviation) or Scott's ($c=1$, $s=$standard deviation), used in nonparametric density estimation [see \cite{Silverman86}], is  
\begin{equation}
  \label{bandwidth}
h_n =  1.2^2 \frac{2\tr(\widehat{\Sigma}_\xn)}{p} \left(n^{-1/(4+p-q)} \right)^2.
\end{equation}
The term $2 \tr(\widehat{\Sigma}_\X)/p$ can be interpreted as the variance  of $\X_i - \X_j$ and $p-q$ is the true dimension $k$.
We use 1.2 as $c$ based on empirical evidence from simulations. 
Since both \eqref{hn} and \eqref{bandwidth} yield satisfactory results, we opted against cross validation for bandwidth selection because of the computational burden involved, and used the bandwidth in \eqref{bandwidth} in simulations and data analyses. 

\section{Optimization Algorithm}\label{Optim}

A Stiefel manifold optimization algorithm is used to obtain  the solution of the sample version of the optimization problem \eqref{optim}. To calculate $\widehat{\V}_{q}$ in \eqref{Vhat}, a curvilinear search is carried out \cite{ZaiwenWen2012,Tagare2011}, which is similar to gradient descent. 
First  an arbitrary starting value $\V^{(0)}$ is selected by drawing a $p \times q$ matrix  from the invariant measure; i.e.,  the distribution that corresponds to the uniform,  on $\spc(p,q)$, see \cite{StatisticsOnManifolds}. The $Q$-component  of the \textsc{QR} decomposition of a $p \times q$ matrix with independent standard normal entries follows the invariant measure \cite{Chikuse1994}. The step-size $\tau > 0$, the step size reduction factor $\gamma \in (0,1)$, and tolerance $\text{tol} > 0$ are fixed at the outset.\\

{\centering
\begin{minipage}{1.1\linewidth}
\begin{algorithm}[H]\label{algo1}
\SetAlgoLined
\KwResult{$\V^{(\text{end})}$}
 Initialize:
 $\V^{(0)}$, $\tau = 1$, $\text{tol} = 10^{-3}$, $\gamma = 0.5$
  $\text{error} = \text{tol} + 1$, $\text{maxit} = 50$, $\text{count}=0$;\\
 \While{$\text{error} > \text{tol}$ $\text{and}$ $\text{count} \leq \text{maxit}$}{
 \begin{itemize}
     \item $\G = \nabla_{\V}L_n(\V^{(j)}) \in \real^{p \times q}$, $\W = \G \V^T - \V \G^T$ 
     \item $\V^{(j+1)} = (\I_p + \tau \W)^{-1}(\I_p - \tau \W)\V^{(j)}$
     \item $\text{error} = \|\V^{(j)}\V^{(j)\T}-\V^{(j+1)}\V^{(j+1)^T}\|/\sqrt{2q}$ 
 \end{itemize}
  \eIf{$L_n(\V^{(j+1)}) > L_n(\V^{(j)}) $}{
 $\V^{(j+1)} \leftarrow \V^{(j)}$; $\tau \leftarrow \tau \gamma$; $\text{error} \leftarrow \text{tol} + 1$
   }{
   $\text{count} \leftarrow \text{count} + 1$\\
   $\tau \leftarrow \frac{\tau}{\gamma}$
   }
}
\caption{Curvilinear search}
\end{algorithm}
\end{minipage}
\par
}

\medskip
Under mild regularity conditions on the objective function, \cite{ZaiwenWen2012} showed  that the sequence generated by the algorithm converges to a stationary point if the Armijo-Wolfe conditions \cite{ArmijoWolfe} are used for determining the stepsize $\tau$. 

The Armijo-Wolfe conditions require the evaluation of the gradient for each potential step size until one is found that fulfills the conditions and the step is accepted, i.e. for the determination of one step size the gradient has to be evaluated multiple times. Since for the conditional variance estimator, the gradient computation incurs the highest computational cost, we use simpler conditions to determine the step size. Specifically,  we simply require the step decrease the objective function, otherwise the step size $\tau$ is decreased by the factor $\gamma \in (0,1)$). These simplified conditions are computationally less expensive and exhibit same behavior as the Armijo-Wolfe conditions in the simulations. Further we capped the maximum number of steps at $\text{maxit} = 50$ steps, since the algorithm  converged in about 10 iterations in all our simulations.

The algorithm  is repeated  for $m$ arbitrary $\V^{(0)}$ starting values drawn from the invariant measure on $\spc(p,q)$. Among those, the value at which $L_n$ in \eqref{LN} is minimal is selected as $\widehat{\V}_{q}$. 

The algorithm requires the computation of the gradient of $L_n(\V)$ in \eqref{LN} or \eqref{wLN}. We compute the gradient  of the objective function for the Gaussian kernel in Theorems~\ref{lemma-one} and~\ref{lemma-two}. 
The Gaussian kernel is the default kernel we use in the implementation of the estimation algorithm in the \texttt{R} code that accompanies this manuscript. 

\begin{thm}\label{lemma-one}
Let $K(z) = \exp{(-z^2/2)}$ be the Gaussian kernel. Then, the gradient of $\tilde{L}_n(\V,\bs_0) $ in \eqref{Ltilde} is given by 
\begin{align*}
\nabla_{\V}\tilde{L}_n(\V,\bs_0) = \frac{1}{h_n^2}\sum_{i=1}^n (\tilde{L}_n(\V,\bs_0) - (Y_i-\bar{y}_1(\V,\bs_0))^2)w_id_i\nabla_{\V}d_i(\V,\bs_0) \in \real^{p \times q},
\end{align*}
and the gradient of $L_n(\V)$ in \eqref{LN} is 
\[
\nabla_{\V}L_n(\V) = \frac{1}{n} \sum_{i=1}^n \nabla_{\V}\tilde{L}_n(\V,\X_i).
\]
with $w_i = {w}(\V,\X_i)$ in \eqref{weights}. 
\end{thm}

The weighted version of conditional variance estimation in Section \ref{weight_section}  is expected to increase the accuracy of the estimator for unevenly spaced data. When \eqref{wLN} and the gradient in \eqref{full}  are used in the optimisation algorithm, we refer to the  estimator as  \textit{weighted conditional variance estimation}. If \eqref{wLN} and the gradient $\sum_{i=1}^n \tilde{w}(\V,\X_i)\nabla_{\V}\tilde{L}_n(\V,\X_i)$ is used; i.e., the first summand  in \eqref{full} is dropped, we refer to it as \textit{partially weighted conditional variance estimation}.
  For both, we replace $G$ in algorithm~\ref{algo1} with the corresponding gradient derived in Theorem~\ref{lemma-two}.

\begin{thm}\label{lemma-two}
Let $K(z) = \exp{(-z^2/2)}$ be the Gaussian kernel. Then, the gradient of $L^{(w)}_n(\V)$ in \eqref{wLN} is given by
\begin{align}
\nabla_{\V}L^{(w)}_n(\V) &=  \sum_{i=1}^n \left(\nabla_{\V}\tilde{w}(\V,\X_i) \tilde{L}_n(\V,\X_i) + \tilde{w}(\V,\X_i)\nabla_{\V}\tilde{L}_n(\V,\X_i)\right), \label{full}
\end{align}
where $\nabla_{\V}\tilde{L}_n(\V,\X_i)$ is given in Theorem~\ref{lemma-one}. Furthermore, 
\[ \nabla_{\V}\tilde{w}(\V,\X_i) = -\frac{1}{h_n^2} \sum_j \left( \frac{K_{j,i}}{\sum_{l,u=1}^n K_{l,u}}d_{j,i}\nabla_{\V} d_{j,i} - \tilde{w}_i \sum_{l,u=1}^n \frac{K_{l,u}}{\sum_{o,s=1}^n K_{o,s}}d_{l,u}\nabla_{\V} d_{l,u} \right)
\]
with $\tilde{w}_i = \tilde{w}(\V,\X_i)$ in \eqref{wtilde}, $K_{j,i} = K(d_j(\V,\X_i)/h_n)$, and $d_{j,i} = d_j(\V,\X_i)$ given in \eqref{distance}.
\end{thm}

\subsection{\texorpdfstring{A study of the behaviour of $L_n(\V)$}{A study of the behaviour of Ln(V)}}\label{ToyExample}
	
We explore how accurately the sample version \eqref{LN} of the objective function estimates the target subspace in an example. We consider a bivariate normal predictor vector, $\X = (X_1,X_2)^T  \sim N(\mathbf{0},\Sigmaxbf)$. We generate the response from $Y = g(\B^T\X) + \epsilon = X_1 + \epsilon$, with $\epsilon \sim N(0,\eta^2)$ independent of $\X$. In this setting, $k = 1$, $\B = (1,0)^T $, $g(z) = z \in \real$ in  model~\eqref{mod:basic}.
With these specifications, \eqref{mu_l} 
becomes
\begin{align}\label{mul}
\mu_l(\V,\bs_0) 
&=  \int_{\real}(\B^T\bs_0 + \B^T \V r)^l f_{\X\mid\X \in \bs_0 +\spn\{\V\}}(r)dr 
\end{align}
Dropping the terms that do not contain $\rs$ in \eqref{density} yields
\begin{gather}
f_{\X\mid\X \in \bs_0 +\spn\{\V\}}(r) \propto f_\X(\bs_0 + \V r) \propto \exp{\left(-\frac{1}{2}(\bs_0 + r\V)^T\Sigmaxbf^{-1}(\bs_0 + r \V)\right)} \notag \\
 \propto \exp{\left(-\frac{1}{2}\left(2r\V^T\Sigmaxbf^{-1}\bs_0 +  r^2\V^T\Sigmaxbf^{-1}\V\right)\right)} 
= \exp{\left(-\frac{1}{2\sigma^2}\left(2r\sigma^2\V^T\Sigmaxbf^{-1}\bs_0 + r^2\right)\right)} \notag \\
\propto \exp{\left(-\frac{1}{2\sigma^2}(r - \alpha)^2\right)}, \label{third}
\end{gather}
where 
$\sigma^2 =1/(\V^T\Sigmaxbf^{-1}\V)$, $ \alpha = -\sigma^2\V^T\Sigmaxbf^{-1}\bs_0$ and the symbol $\propto$ stands for proportional to. Letting $\psi(z) $ denote the density of a standard normal variable, \eqref{third} obtains
\begin{align}\label{fcond}
f_{\X\mid\X \in \bs_0 +\spn\{\V\}}(r) = 
\frac{1}{\sigma}\psi\left(\frac{r- \alpha}{\sigma}\right) 
\end{align}
for $\V, \bs_0 \in \real^{2 \times 1}$.
Inserting \eqref{fcond} in \eqref{mul} yields  
\begin{gather*}
\int_{\real} (\B^T\bs_0 + \B^T\V r)^l \frac{1}{\sigma}\psi\left(\frac{r- \alpha}{\sigma}\right) d r 
 = \begin{cases}
	\B^T\bs_0 + \B^T\V\alpha &  l = 1\\
	(\B^T\bs_0)^2 + 2(\B^T\bs_0)(\B^T\V)\alpha + (\B^T\V)^2(\sigma^2 + \alpha^2) & l = 2 \\
	\end{cases}
\end{gather*}
 
Using 
\eqref{LtildeVs0},~\eqref{Lvs} and \eqref{objective}, yields $\tilde{L}(\V,\bs_0) = \mu_2(\V,\bs_0) - \mu_1(\V,\bs_0)^2 +\eta^2 = (\B^T\V)^2\sigma^2 + \eta^2$, so that 

\begin{equation}
L(\V) = \E\left(\tilde{L}(\V,\X)\right) = (\B^T\V)^2\sigma^2 + \eta^2 = \frac{(\B^T\V)^2}{\V^T\Sigmaxbf^{-1}\V} + \eta^2 \label{toyexample}
\end{equation}

From \eqref{toyexample} we can easily  see that  $L(\V)$ attains its minimum at $\V \perp \B$. Also, if $\Sigmaxbf=\I_2$, the maximum of $L(\V)$ is attained at  $\V = \B$. 
To visualize the behavior of $\tilde{L}_n(\V)$ as the sample size increases, we parametrize $\V$ by $\V(\theta) = (\cos(\theta),\sin(\theta))^T$,  $\theta \in [0,\pi]$. Since $\B = (1,0)^T$,  the minimum of $\tilde{L}(\V)$ is at $\V(\pi/2) = (0,1)^T$ , which is orthogonal to $\B$.

The true $L(\V(\theta))$ and its estimates $L_n(\V(\theta))$ are plotted for samples of different sizes $n$ in Figure~\ref{Lvplot}. $L_n(\V(\theta))$ approximates $L(\V)$ fast and attains its minimum at the same value as $L(\V)$ even for $n= 10$. 

As an aside, we note that assumption (A.4) is violated in this example, which suggests that the proposed estimator of conditional variance estimation may apply under weaker assumptions.   

\begin{figure}[htbp] 
\centering
\includegraphics[scale=0.4]{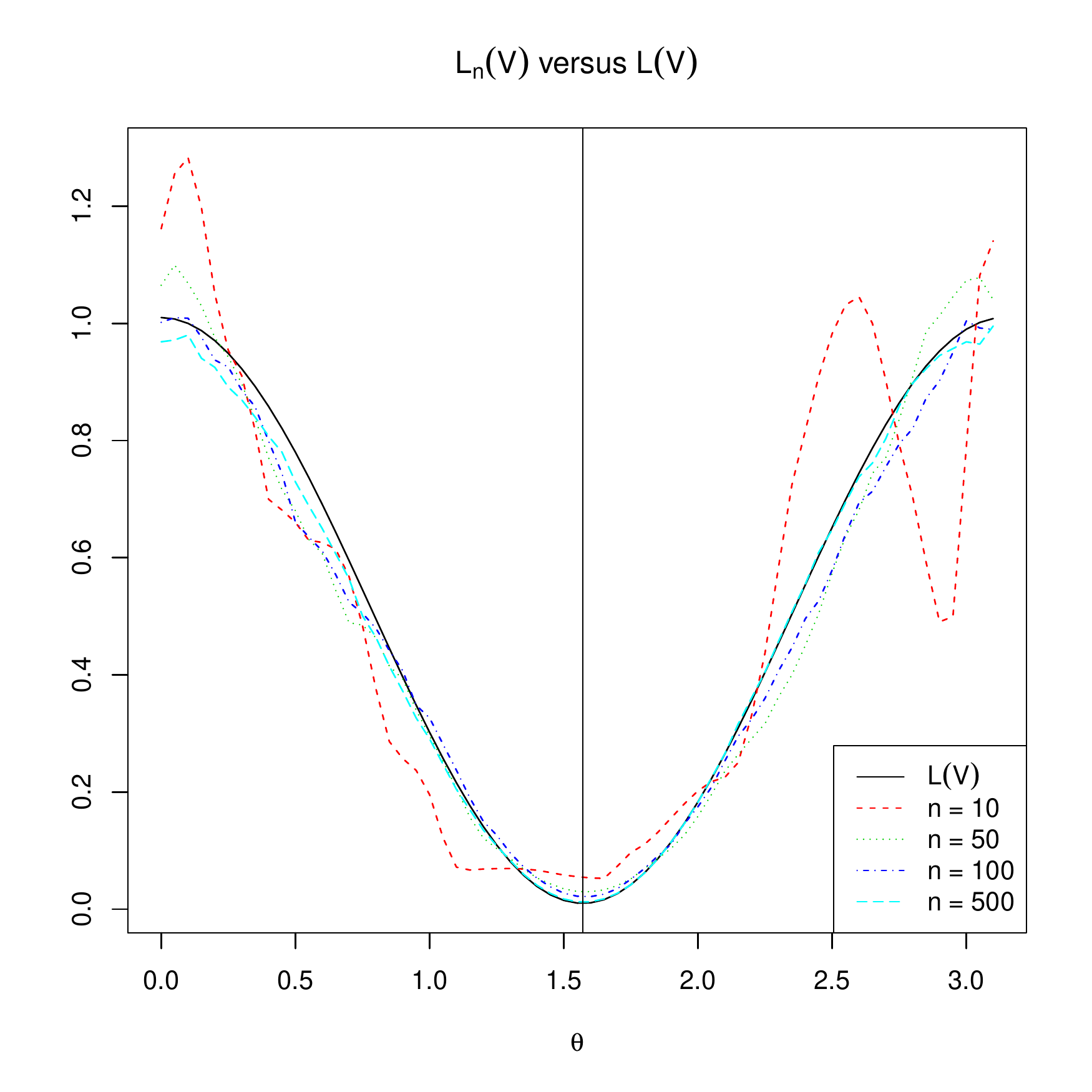}
\caption{Solid black line is $L(\V(\theta)) = \cos(\theta)^2 + 0.1^2$, colored is $L_n(\V(\theta))$, $\theta \in [0, \pi]$, $n=10,50,100,500$. The vertical black line is at $\theta=\pi /2$
}
\label{Lvplot}
\end{figure}

\section{Simulation studies}\label{SimStudy}
We compare the estimation accuracy of conditional variance estimation  with the forward model based sufficient dimension reduction methods, mean outer product gradient estimation (\texttt{meanOPG}), mean minimum average variance estimation (\texttt{meanMAVE}) \cite{MAVEpackage},  refined outer product gradient (\texttt{rOPG}), refined minimum average variance estimation (\texttt{rmave}) \cite{Xiaetal2002, Li2018}, and principal Hessian directions (\texttt{pHd}) \cite{Li1992, CookLi2002}, and the inverse regression based methods, sliced inverse regression (\texttt{SIR}) \cite{Li1991} and sliced average variance estimation (\texttt{SAVE}) \cite{CookWeisberg1991}. The dimension $k$ is assumed to be known throughout. 

We report results for conditional variance estimation  using the ``plug-in'' bandwidth in \eqref{bandwidth} 
and three different conditional variance estimation versions, \texttt{CVE}, \texttt{wCVE}, and \texttt{rCVE}. \texttt{CVE} is obtained by using $m = 10$ arbitrary starting values in the optimization algorithm and  optimizing \eqref{LN} as described in Section~\ref{Optim}. \texttt{rCVE}, or \textit{refined weighted CVE}, is obtained by setting the starting value $\V^{(0)}$ at the optimizer of \texttt{CVE}, and using \eqref{wLN} in the optimization algorithm in Section~\ref{Optim} with the partially weighted gradient as described in Section~\ref{weight_section}. \texttt{wCVE}, or \textit{weighted CVE}, is obtained by optimizing \eqref{wLN} with partially weighted gradient as described in Sections~\ref{weight_section} and \ref{Optim}. Methods \texttt{rOPG} and \texttt{rmave} refer to the original refined outer product gradient and refined minimum average variance estimation algorithms published in \cite{Xiaetal2002}. They are implemented using the \texttt{R} code in \cite{Li2018} with number of iterations $\text{nit}=25$, since the algorithm is seen to converge by 25. The \texttt{dr} package is used for the \texttt{SIR}, \texttt{SAVE} and \texttt{pHd} calculations, and the \texttt{MAVE} package for mean outer product gradient estimation (\texttt{meanOPG}) and mean minimum average variance estimation (\texttt{meanMAVE}). The source code for conditional variance estimation  can be downloaded from  \url{https://git.art-ist.cc/daniel/CVE}. 

Table~\ref{tab:mod} lists the seven models (M1-M7) we consider. 
Throughout, we set $p=20$,  $\bb_1 = (1,1,1,1,1,1,0, ...,0)^T/\sqrt{6}$, $\bb_2 = (1,-1,1,-1,1,-1,0,...,0)^T/\sqrt{6} \in \real^p$ for M1-M5. For M6, $\bb_1 = \eb_1, \bb_2 = \eb_2$ and $\bb_3 = \eb_p$, and for M7 $\bb_1,\bb_2,\bb_3$ are the same as  in M6 and $\bb_4 = \eb_3$, where $\eb_j$ denotes the $p$-vector with $j$th  element equal to 1 and all others are 0. The error term $\epsilon$ is independent of $\X$ for all models. In M2, M3, M4, M5 and M6, $\epsilon \sim N(0,1)$. For M1 and M7, $\epsilon$ has a generalized normal distribution $GN(a,b,c)$ with densitiy $f_{\epsilon}(z) = c/(2b\Gamma(1/c))\exp((|z-a|/b)^c)$, see \cite{gnorm} with location 0 and shape-parameter 0.5  for M1, and shape-parameter 1 for M7 (Laplace distribution). For both the scale-parameter is chosen such that $\var(\epsilon) = 0.25$.

\begin{table}[!htbp]
\centering
\caption{Models}
\vspace{0.05in}
\resizebox{\columnwidth}{!}{%
\begin{tabular}{lccccc}
		\toprule
Name & Model & $\X$ distribution & $\epsilon$ distribution& $k$ & $n$ \\ \midrule
M1& $Y = \cos(\bb_1^T\X) + \epsilon$ & $\X \sim N_p(\0,\Sigmabf)$& $GN(0,\sqrt{1/2},0.5)$ & 1 & 100\\
M2 &$Y = \cos(\bb_1^T\X) + 0.5\epsilon$  & $ \X \sim  \lambda Z \1_{p} + N_p(\0,\I_{p})$& $N(0,1)$ & 1 & 100\\
M3& $Y = 2\log(|\bb_1^T\X|+2)+ 0.5\epsilon$&$\X \sim N_p(\0,\I_p)$& $N(0,1)$ & 1 & 100\\
M4& $Y = (\bb_1^T\X)/(0.5 +(1.5 + \bb_2^T\X)^2) + 0.5\epsilon$&$\X \sim N_p(0,\Sigmabf)$& $N(0,1)$ & 2 & 200\\
M5 & $Y = \cos(\pi \bb_1^T\X)(\bb_2^T\X + 1)^2 + 0.5\epsilon$&$\X \sim U([0,1]^p)$ & $N(0,1)$& 2 & 200\\
M6 &$Y = (\bb_1^T\X)^2 + (\bb_2^T\X)^2 + (\bb_3^T\X)^2 + 0.5\epsilon$ & $\X \sim N_p(\0,\I_p)$& $N(0,1)$ & 3 & 200\\
M7 &$Y = (\bb_1^T\X)(\bb_2^T\X)^2 + (\bb_3^T\X)(\bb_4^T\X) + \epsilon$ & $\X \sim t_3(\I_p)$& $GN(0,\sqrt{1/\Gamma(6)},1)$ & 4 & 400\\
 		\bottomrule
	\end{tabular}%
	}
	\label{tab:mod}%
\end{table}

The variance-covariance structure of $\X$ in models M1 and M4 satisfies $\Sigmabf_{i,j} = 0.5^{|i-j|}$ for $i,j=1,\ldots,p$. In M5, $\X$ is uniform  with independent entries on the $p$-dimensional hyper-cube. In M7, $\X$ is multivariate $t$-distributed with 3 degrees of freedom. The link functions of M4 and M7 are studied in  \cite{Xiaetal2002}
, but we use $p=20$ instead of 10 and a non identity covariance structure for M4 and the $t$-distribution instead of normal for M7. 
In M2, 
$Z \sim 2\text{Bernoulli}(p_{\text{mix}}) - 1 \in \{-1,1\}$, 
where $\1_q = (1,1,...,1)^T\in \real^q$, mixing probability $p_{\text{mix}} \in [0,1]$ and dispersion parameter $\lambda > 0$. 
For $0 < p_{\text{mix}} <1$, $\X$ has a mixture normal distribution, where $p_{\text{mix}}$ is the relative mode height and $\lambda$ is a measure of mode distance.

We set $q = p - k$ and generate $r=100$ replications of models M1 - M7. We estimate $\B$ using the ten sufficient dimension reduction methods. The accuracy of the estimates is assessed using  $err= \|\Pbf_\B - \Pbf_{\widehat{\B}}\|/\sqrt{2k}$, which lies in the interval $[0,1]$. The factor $\sqrt{2k}$ normalizes the distance, with values closer to zero indicating better agreement and values closer to one indicating strong disagreement 
, specifically,  $\|\Pbf_\B - \Pbf_{\widehat{\B}}\|^2 \leq 2k$.

In Table~\ref{tab:summary} the mean and standard deviation of $err$ for M1 - M7 are reported. In particular, for M2, $p_{mix}=0.3$ and $\lambda = 1$. The smallest error values are boldfaced. In models M1, M2 and M3, the conditional variance estimator  is the best performer, with its refined version as close second. In M4, M5 and M6, any of the four versions of MAVE performs better than the CVE. For model M7 the results of \texttt{rOPG} and \texttt{rmave} are not reported because the code frequently produces an error message that a matrix is not invertible. Among the rest, the weighted version of CVE, wCVE, attains the minimum error.  

Sliced inverse regression (\texttt{SIR}) and sliced average variance estimation (\texttt{SAVE}) are not competitive throughout our experiments. Sliced inverse regression (\texttt{SIR}), in particular, is expected to fail in models M1-M3, and M6 since  $\E(Y\mid\X)$ is even.

 In Figure~\ref{fig:M4}, box-plots for all combinations of $p_{\text{mix}} \in \{0.3,0.4,0.5\}$ and $\lambda \in \{0,0.5,1,1.5\}$ are presented. The reference methods are restricted to \texttt{meanOPG} and \texttt{meanMAVE}, since the others are not competitive. Conditional variance estimation performs better than all competing methods and  is the only method with consistently smaller errors when the two modes are further apart ($\lambda \geq 1$) regardless of the mixing probability $p_{\text{mix}}$. The performance of both \texttt{meanOPG} and \texttt{meanMAVE}  worsens as one moves from left to right  row-wise. The mixing probability, $p_{\text{mix}}$, has no noticeable effect on the performance of any method; i.e., the plots are very similar column-wise. In sum, \texttt{meanMAVE}'s performance deteriorates as the bimodality of the predictor distribution becomes more distinct. In contrast, conditional variance estimation is unaffected.
and appears to have an advantage over \texttt{meanMAVE} when the predictors have mixture distributions, the link function is even about the midpoint of the two modes, and $\B$ is not orthogonal to the line connecting the two modes. Conditional variance estimation  is the only method that estimates the mean subspace reliably in model M2 ($err$ $\approx 0.4$ to $0.5$), whereas \texttt{meanMAVE} misses it completely ($err$ $\approx 1$). 
These results indicate that conditional variance estimation is often approximately on par, and can perform much better than \texttt{meanMAVE} depending on the predictor distribution and the link function.

\begin{table}[!htbp]
\centering
\caption{Mean and standard deviation of estimation errors}
\vspace{0.05in}
	\label{tab:summary}
\resizebox{\columnwidth}{!}{%
\begin{tabular}{lcccccccccc}
		\toprule
Model & CVE &wCVE&rCVE & meanOPG & rOPG& meanMAVE& rmave& pHd& sir & save \\ \midrule
$M1$\\
\qquad mean& {\bf0.3827}& 0.4414 &0.4051&0.6220&0.9876&0.5099&0.9840&0.8278&0.9875&0.9788\\
\qquad sd& 0.1269& 0.1595 &0.1329&0.1879&0.0223&0.1800&0.0295&0.1206&0.0243&0.0334\\

$M2$\\
\qquad mean& {\bf0.4572}&0.4992&0.4658&0.8987&0.9332&0.8905&0.9242&0.9000&0.9783&0.9781\\
\qquad sd& 0.1038& 0.1524& 0.0989&0.0908&0.0683&0.0983&0.0897&0.0735&0.0278&0.0318\\

$M3$ \\
\qquad mean& {\bf0.6282}& 0.7509 &0.6371&0.7847&0.9644&0.7576&0.9674&0.6964&0.9647&0.9519\\
\qquad sd& 0.2354& 0.2262 &0.2181&0.2201&0.0667&0.2435&0.0609&0.1626&0.0587&0.0650\\

$M4$\\
\qquad mean& 0.5663&0.5897&0.5554&0.4071&0.4026&0.4361&{\bf0.3905}&0.7772&0.5824&0.9727\\
\qquad sd& 0.1239& 0.1246&0.1298&0.0814&0.0609&0.0997&0.0584&0.0662&0.0951&0.0202\\

$M5$\\
\qquad mean& 0.4429& 0.5604&0.4779&0.4058&{\bf0.3737}&0.3929&0.3750&0.7329&0.6374&0.9730\\
\qquad sd& 0.0891& 0.1233&0.0976&0.1022&0.0680&0.0894&0.0871&0.0832&0.0968&0.0186\\

$M6$\\
\qquad mean& 0.3828&0.3027&0.3230&0.1827&0.4632&{\bf0.1656}&0.4863&0.4978&0.9129&0.8236\\
\qquad sd& 0.1006& 0.0748&0.1098&0.0289&0.1717&0.0252&0.1676&0.0601&0.0420&0.0518\\

$M7$\\
\qquad mean& 0.6856& {\bf0.5050}&0.5651&0.5694&NA&0.5482&NA&0.8536&0.8133&0.8699\\
\qquad sd& 0.0588& 0.0862&0.0879&0.1122&NA&0.1271&NA&0.0354&0.0341&0.0342\\
 		\bottomrule
	\end{tabular}%
	}
\end{table}

\begin{figure}[!htbp] 
		\centering
		\includegraphics[scale=.32]{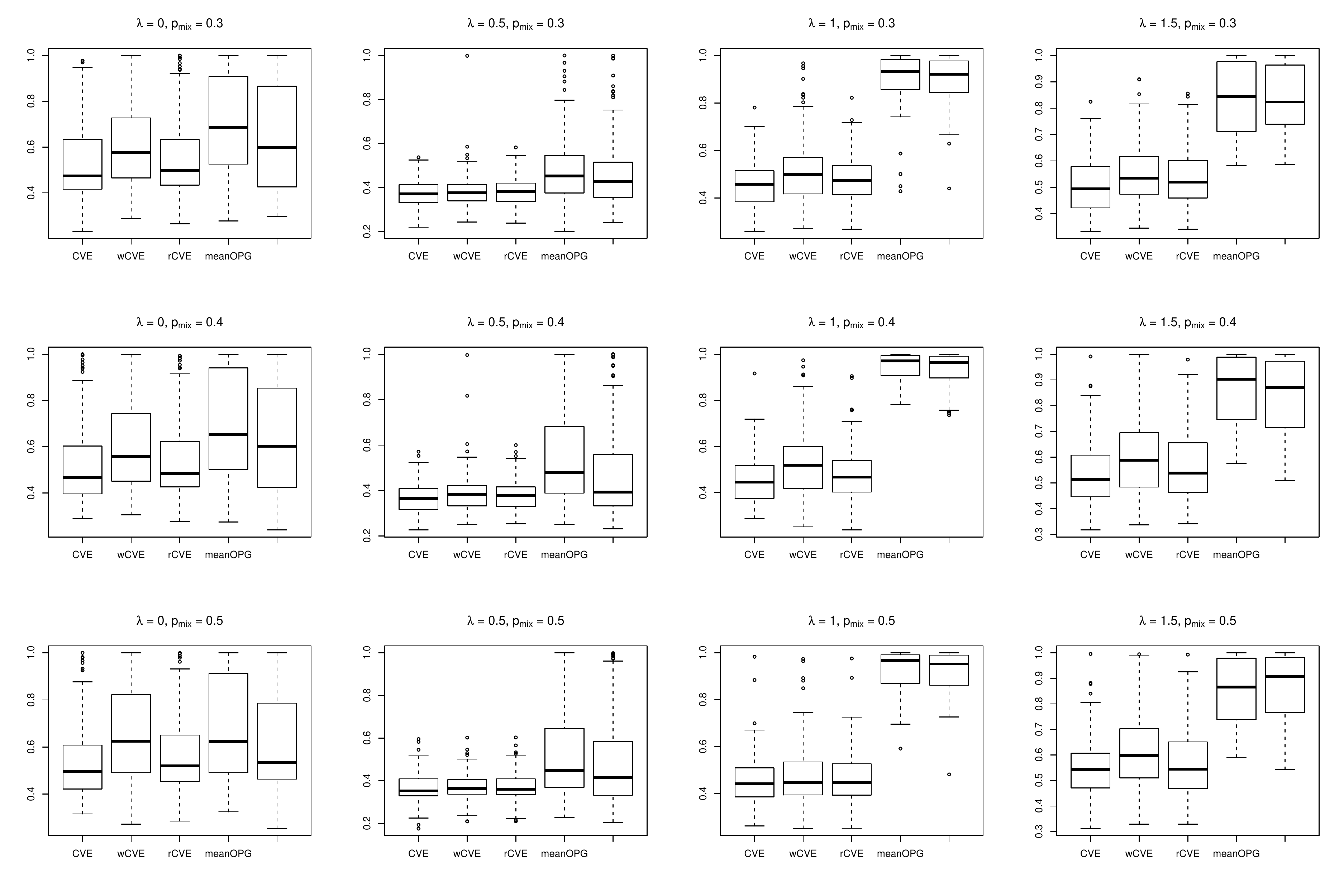}
		\caption{M2, $p = 20, n = 100$}
		\label{fig:M4}
\end{figure}
Furthermore we estimate the dimension $k$ via cross-validation, following the approach in \cite{Xiaetal2002} 
, with 
\begin{align}
\hat{k} &= \argmin_{l=1,...,p}CV(l)= \argmin_{l=1,...,p} \frac{\sum_i (Y_i - \hat{g}^{-i}(\widehat{\B}_{l}^T\X_i))^2}{n},\label{estdim}
\end{align}
where $\hat{g}^{-i}(\cdot)$ is computed from the data $(Y_j,\widehat{\B}_{l}^T\X_j)_{j=1,...,n; j \neq i}$ using multivariate adaptive regression splines \cite{mars} in the \texttt{R}-package \texttt{mda}
, and $\widehat{\B}_{l} =\widehat{\V}_{p-l}^\perp$ is any basis of the orthogonal complement of $\widehat{\V}_{p-l}= \argmin_{\V \in \spc(p,p-l)}L_n(\V)$.
For a given $l$,  we calculate $\widehat{\B}_l$ from the whole data set and predict $Y_i$ by $\hat{Y}_{i,l} = \hat{g}^{-i}(\widehat{\B}_l^T\X_i)$. For $l=p$, $\widehat{\B}_{p} =\I_p$. The results for the seven models are reported in Table~\ref{pred_dim}. The CVE based dimension estimation is the most accurate in models M1, M2, M3, and M6 and differs slightly from that of MAVE in M7. MAVE performs better in M4 and M5, completely misses the true dimension in M2 and misses it most of the time in M3. Thus, the dimension estimation performance of CVE and MAVE agrees with the estimation accuracy of the true subspace in Table~\ref{tab:summary}, CVE  estimates the dimension more accurately even in model M6, where it exhibits worse subspace estimation performance, and overall appears to be more accurate. 

\begin{table}[!htbp]
    \centering
    \caption{Number of times dimension $k$ is correctly estimated in  $100$ replications}
    \vspace{0.05in}
\begin{tabular}{l|ccccccc}
		\toprule
 &M1 & M2 & M3& M4 & M5&M6&M7\\ \midrule
 CVE &83&41&88&62&46&74&19\\
MAVE &67&0&14&76&60&57&21\\ 
 		\bottomrule
	\end{tabular}%
	\label{pred_dim}
\end{table}
	
We carried out many simulation experiments for an array of combinations of link functions, sufficient reduction matrices $\B$ and their ranks, as well as predictor and error distributions. All reported and unreported results indicate that the difference in performance of the two methods, CVE and mean MAVE, can be attributed to both the form of the link function and the marginal predictor distribution. We observed that when the link function had a bounded first order derivative, CVE often outperformed mean MAVE across predictor distributions. In the opposite case, MAVE performed mostly better. Also, when the predictors have a bimodal distribution with well separated modes and the link function is even, regardless of whether its  derivative is bounded, CVE outperforms mean MAVE. In the other settings for the generated data, both methods were roughly on par.

\section{Real Data Analyses}\label{real_data}
Three data sets are analyzed:  the \textit{Hitters} data in the \texttt{R} package \texttt{ISLR}, which was also analyzed by \cite{Xiaetal2002}, the \textit{Boston Housing} data  in the \texttt{R} package \texttt{mlbench}, and the \textit{Concrete} data from the \texttt{MAVE} package.
The reference method is \texttt{meanMAVE} from the \texttt{MAVE} package in \texttt{R} and the \texttt{CVE} is calculated using $m = 50$ and $\text{maxit} = 10$ in the optimization algorithm \ref{algo1} in Section~\ref{Optim}. The estimation of the dimension is based on \eqref{estdim} in Section~\ref{SimStudy}.


Following \cite{Xiaetal2002}, we remove 7 outliers from the \textit{Hitters} data set leading to a sample size of 256. The response is $Y = \log(\text{salary})$ and the 16 continuous predictors are  the game statistics of players in the Major League Baseball league in the seasons 1986 and 1987. 
Further information can be found in \url{https://www.rdocumentation.org/packages/ISLR/versions/1.2/topics/Hitters}.

The \textit{Boston Housing} data set contains 506 census tracts on 14 variables  from the 1970 census. The response is \texttt{medv}, the median value of owner-occupied homes in USD 1000's. The factor variable \texttt{chas} is removed from the data set for the analysis so that the response is modeled by the remaining 12 continous predictors. The description of the variables can be found in \url{https://www.rdocumentation.org/packages/mlbench/versions/2.1-1/topics/BostonHousing}.  

The \textit{Concrete} data set contains 1030 instances on 9 continuous variables 
The response is concrete compressive strength. 
Concrete strength is very important in civil engineering and is a highly nonlinear function of age and ingredients. The description of the variables can be found in \url{https://www.rdocumentation.org/packages/MAVE/versions/1.3.10/topics/Concrete}.

For all three data sets we standardize both the predictors and the response by subtracting the mean and rescaling column-wise so that each variable has unit variance. 
The data sets are analyzed using 10 fold cross-validation to calculate an unbiased estimate of the  prediction error \cite{crossvalidation} for our method , \texttt{CVE}, and its main competitor \texttt{meanMAVE} using the \texttt{MAVE}  package. 
The dimension for each method is estimated with \eqref{estdim} on the trainings set and we then fit a forward  regression model on the training set replacing the original with the reduced predictors using multivariate adaptive regression splines \cite{mars} using the \texttt{R} package \texttt{mda} and calculate the prediction error on the test set for both methods. 
The dimension estimates of \texttt{CVE} and \texttt{MAVE} mostly disagree.

 The mean and standard deviation of the 10-fold cross-validation prediction errors are reported in Table~\ref{table:datasets}. Since the response is standardized, the values in Table~\ref{table:datasets} are bounded between 0 and 1, with smaller values indicating better predictive performance.  \texttt{CVE} performs slightly worse than mean \texttt{MAVE} in the \textit{Hitters} data set, slightly better in the \textit{Boston Housing} and 
 better in the \textit{Concrete} data set analysis.

 \begin{table}[!htbp]
\centering
\caption{Mean and standard deviation (in parenthesis) of standardized out of sample prediction errors for the three data sets}
\vspace{0.05in}
{
\begin{tabular}{l|ccc}
		\toprule
Method& Hitters &  Housing &Concrete\\
\bottomrule
CVE  & 0.216       &  0.260   &0.361  \\
 & (0.101) & (0.331) & (0.206) \\
MAVE  & 0.203       &   0.299&0.417 \\
& (0.083) & (0.382) & (0.348) \\
\bottomrule

	\end{tabular}%
	}
	\label{table:datasets}%
\end{table}

\subsection{\texorpdfstring{Hitters Data Analysis as in \cite{Xiaetal2002}}{Hitters Data Analysis}}\label{hitters}
Additionally, 
we reconstruct the analysis of the \textit{Hitters} data in \cite{Xiaetal2002}, which does not account for the out-of-sample prediction error as in Section~\ref{real_data} but uses the whole sample for estimation of $\B$ and its rank. Only the dimension $k$ is estimated with leave-one-out cross validation.

Table~\ref{tab:ex1} reports the average cross validation mean squared error $CV(k)$ in \eqref{estdim}  using  the whole data set over $k=1,\ldots,5$.
Both conditional variance estimation 
and mean minimum average variance estimation estimate the dimension to be 2. 
\begin{table}[!htbp]
    \centering
    \caption{Mean cross-validation error}
    \vspace{0.05in}

\begin{tabular}{l|ccccc}
		\toprule
	$k$ &1 & 2 & 3& 4 & 5\\ \midrule
 CVE &0.308&0.218&0.275&0.327&0.371\\
MAVE &0.370&0.277&0.339&0.413&0.440\\ 
 		\bottomrule
	\end{tabular}%
	\label{tab:ex1}
\end{table}

We plot the response against the estimated directions in Figure~\ref{Fig 6}. 
\begin{figure}[htbp] 
	\centering
	\includegraphics[width=0.7\textwidth]{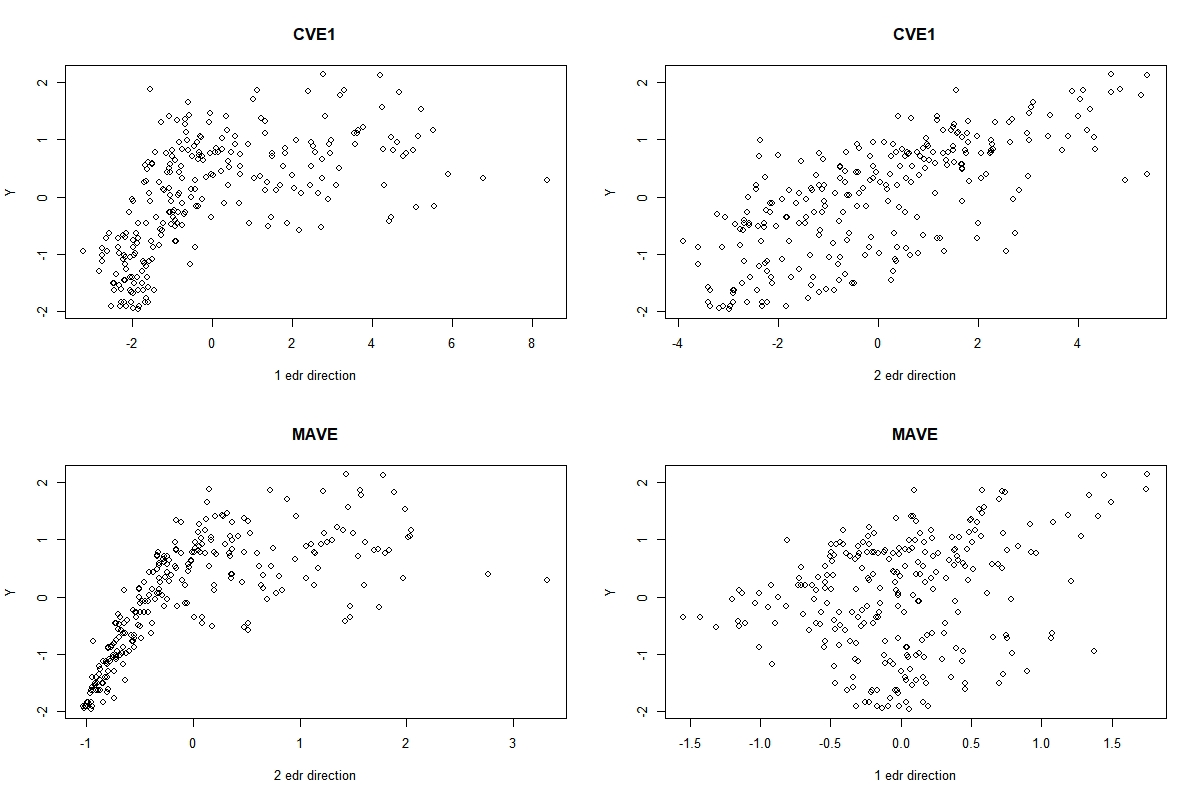}
	\caption{ $Y$ against $\widehat{\bb}_1^T\X$ and $\widehat{\bb}_2^T\X$}
	\label{Fig 6}
\end{figure}
Both exhibit the same pattern: the response appears to be linear in one direction and quadratic in the second. The difference is that the linear pattern is clearer in the second CVE direction and the quadratic pattern  exhibits increasing variance in the first MAVE direction. 

Based on the scatterplots in Figure~\ref{Fig 6}, we fit the same models for both. 
For conditional variance estimation,  the fitted regression is 
\begin{equation} \label{regCVE}
\hat{Y} = 0.39578 + 0.33724 (\widehat{\bb}_1^T\X) - 0.08066 (\widehat{\bb}_1^T\X)^2 + 0.29126 (\widehat{\bb}_2^T\X)  
\end{equation}
with $R^2 = 0.7975$, and
for minimum average variance estimation
\begin{equation}\label{regMAVE}
\hat{Y} = 0.39051 + 1.32529(\widehat{\bb}_1^T\X) - 0.55328(\widehat{\bb}_1^T\X)^2 + 0.49546 (\widehat{\bb}_2^T\X) 
\end{equation}
with $R^2 = 0.7859$.
Both  models \eqref{regCVE} and \eqref{regMAVE} have about the same fit as measured by $R^2$. The in sample performance of the two methods is practically the same  for the \texttt{Hitters} data. 	

\section{Discussion}\label{discussion}

In this paper the novel conditional variance estimator (CVE) for the mean subspace is introduced. We  present its geometrical and theoretical foundation, show its consistency and  propose an estimation algorithm with assured convergence.  CVE requires the forward model \eqref{mod:basic}, $Y=g(\B^T\X) +\epsilon$, holds and weak assumptions on the response and the covariates.

Minimum average variance estimation (MAVE) \cite{Xiaetal2002} is the only other sufficient dimension reduction method based on the forward model \eqref{mod:basic}.  It estimates the sufficient dimension reduction  targeting both the reduction and the link function $g$ in \eqref{mod:basic}. CVE targets only the reduction and does not require estimation of the link function, which may explain why it has an advantage over MAVE in some regression settings. For example,  CVE exhibits similar performance across different link functions (cos, exp, etc) for fixed $\lambda$,  whereas the performance of MAVE  is very uneven for model M2 in Section~\ref{SimStudy}. CVE is more accurate than MAVE when the link function is even and the predictor distribution is bimodal throughout our simulation studies. Moreover, CVE does not require the inversion of the predictor covariance matrix and can be applied to regressions with $p \approx n$ or $p > n$. 

The theoretical challenge in deriving the statistical properties of conditional variance estimation arises from the novelty of its definition that involves random non i.i.d. weights that depend on the parameter to be estimated. 


\begin{thebibliography}{10}

\bibitem{AdragniCook2009}
{Kofi P.} Adragni and {R. Dennis} Cook.
\newblock Sufficient dimension reduction and prediction in regression.
\newblock {\em Philosophical Transactions of the Royal Society A: Mathematical,
  Physical and Engineering Sciences}, 367(1906):4385--4405, 11 2009.

\bibitem{Takeshi}
Takeshi Amemiya.
\newblock {\em Advanced Econometrics}.
\newblock Harvard university press, 1985.

\bibitem{Bernstein}
S.~N. Bernstein.
\newblock {\em Theory of Probability}.
\newblock Moscow, 1927.

\bibitem{Boothby}
W.~M. Boothby.
\newblock {\em An Introduction to Differentiable Manifolds and Riemannian
  Geometry}.
\newblock Academic Press, 2002.

\bibitem{BuraDuarteForzani2016}
Efstathia Bura, Sabrina Duarte, and Liliana Forzani.
\newblock Sufficient reductions in regressions with exponential family inverse
  predictors.
\newblock {\em Journal of the American Statistical Association},
  111(515):1313--1329, 2016.

\bibitem{BuraForzani2015}
Efstathia Bura and Liliana Forzani.
\newblock Sufficient reductions in regressions with elliptically contoured
  inverse predictors.
\newblock {\em Journal of the American Statistical Association},
  110(509):420--434, 2015.

\bibitem{Chikuse1994}
Yasuko Chikuse.
\newblock {\em Invariant measures on Stiefel manifolds with applications to
  multivariate analysis}, volume Volume 24 of {\em Lecture Notes--Monograph
  Series}, pages 177--193.
\newblock Institute of Mathematical Statistics, Hayward, CA, 1994.

\bibitem{StatisticsOnManifolds}
Yasuko Chikuse.
\newblock {\em Statistics on Special Manifolds}.
\newblock Springer-Verlag New York, New York, 2003.

\bibitem{Cook1998}
Dennis~R. Cook.
\newblock {\em Regression Graphics: Ideas for studying regressions through
  graphics}.
\newblock Wiley, New York, 1998.

\bibitem{CookForzani2008}
R.~D. Cook and L.~Forzani.
\newblock Principal fitted components for dimension reduction in regression.
\newblock {\em Statistical Science}, 23(4):485--501, 2008.

\bibitem{Cook2007}
R.~Dennis Cook.
\newblock Fisher lecture: Dimension reduction in regression.
\newblock {\em Statist. Sci.}, 22(1):1--26, 02 2007.

\bibitem{CookForzani2009}
R.~Dennis Cook and Liliana Forzani.
\newblock Likelihood-based sufficient dimension reduction.
\newblock {\em Journal of the American Statistical Association},
  104(485):197--208, 3 2009.

\bibitem{CookLi2004}
R.~Dennis Cook and Bing Li.
\newblock Determining the dimension of iterative hessian transformation.
\newblock {\em Ann. Statist.}, 32(6):2501--2531, 12 2004.

\bibitem{CookWeisberg1991}
R.~Dennis Cook and Sanford Weisberg.
\newblock Sliced inverse regression for dimension reduction: Comment.
\newblock {\em Journal of the American Statistical Association},
  86(414):328--332, 1991.

\bibitem{CookLi2002}
R.Dennis Cook and Bing Li.
\newblock Dimension reduction for conditional mean in regression.
\newblock {\em Ann. Statist.}, 30(2):455--474, 04 2002.

\bibitem{regProb}
Arnold~M. Faden.
\newblock The existence of regular conditional probabilities: Necessary and
  sufficient conditions.
\newblock {\em The Annals of Probability}, 13(1):288--298, 1985.

\bibitem{mars}
Jerome~H. Friedman.
\newblock Multivariate adaptive regression splines.
\newblock {\em The Annals of Statistics}, 19(1):1--67, 1991.

\bibitem{Hansen2008}
Bruce~E. Hansen.
\newblock Uniform convergence rates for kernel estimation with dependent data.
\newblock {\em Econometric Theory}, 24:726–748, 2008.

\bibitem{HarroHeuser}
H.~Heuser.
\newblock {\em Analysis 2, 9 Auflage}.
\newblock Teubner, 1995.

\bibitem{regProb2}
Alan~F. Karr.
\newblock {\em Probability}.
\newblock Springer Texts in Statistics. Springer-Verlag New York, 1993.

\bibitem{Leaoetal2004}
D.~Leao~Jr., M.~Fragoso, and P.~Ruffino.
\newblock Regular conditional probability, disintegration of probability and
  radon spaces.
\newblock {\em {Proyecciones (Antofagasta)}}, 23:15 -- 29, 05 2004.

\bibitem{Li2018}
Bing Li.
\newblock {\em Sufficient dimension reduction: methods and applications with
  R}.
\newblock CRC Press, Taylor \& Francis Group, 2018.

\bibitem{Li1991}
K.~C. Li.
\newblock Sliced inverse regression for dimension reduction.
\newblock {\em Journal of the American Statistical Association},
  86(414):316--327, 1991.

\bibitem{Li1992}
Ker-Chau Li.
\newblock On principal hessian directions for data visualization and dimension
  reduction: Another application of stein's lemma.
\newblock {\em Journal of the American Statistical Association},
  87(420):1025--1039, 1992.

\bibitem{MaZhu2013}
Yanyuan Ma and Liping Zhu.
\newblock A review on dimension reduction.
\newblock {\em International Statistical Review}, 81(1):134--150, 4 2013.

\bibitem{gnorm}
Saralees Nadarajah.
\newblock A generalized normal distribution.
\newblock {\em Journal of Applied Statistics}, 32(7):685--694, 2005.

\bibitem{ArmijoWolfe}
J.~Nocedal and S.~Wright.
\newblock {\em Line Search Methods}, pages 30--65.
\newblock Springer New York, New York, NY, 2006.

\bibitem{Parzen1961}
E~Parzen.
\newblock On estimation of a probability density function and mode.
\newblock {\em The Annals of Mathematical Statistics}, 33(3):1065--1076, 1961.

\bibitem{Silverman86}
B.~W. Silverman.
\newblock {\em Density Estimation for Statistics and Data Analysis}.
\newblock Chapman \& Hall, London, 1986.

\bibitem{crossvalidation}
M.~Stone.
\newblock Cross-validatory choice and assessment of statistical predictions.
\newblock {\em Journal of the Royal Statistical Society: Series B
  (Methodological)}, 36(2):111--133, 1974.

\bibitem{Tagare2011}
Hemant~D. Tagare.
\newblock Notes on optimization on stiefel manifolds, January 2011.

\bibitem{MAVEpackage}
Hang Weiqiang and Xia Yingcun.
\newblock {\em MAVE: Methods for Dimension Reduction}, 2019.
\newblock R package version 1.3.10.

\bibitem{ZaiwenWen2012}
Zaiwen Wen and Wotao Yin.
\newblock A feasible method for optimization with orthogonality constraints.
\newblock {\em Mathematical Programming}, 142:397–434, 2013.

\bibitem{multiIndexModel}
Yingcun Xia.
\newblock A multiple-index model and dimension reduction.
\newblock {\em Journal of the American Statistical Association},
  103(484):1631--1640, 2008.

\bibitem{Xiaetal2002}
Yingcun Xia, Howell Tong, W.~K. Li, and Li-Xing Zhu.
\newblock An adaptive estimation of dimension reduction space.
\newblock {\em Journal of the Royal Statistical Society: Series B (Statistical
  Methodology)}, 64(3):363--410, 2002.

\bibitem{CompactAssumption}
Xiangrong Yin, Bing Li, and R.~Cook.
\newblock Successive direction extraction for estimating the central subspace
  in a multiple-index regression.
\newblock {\em Journal of Multivariate Analysis}, 99:1733--1757, 09 2008.

\end{thebibliography}
\nocite{*} 

\section{Appendix}
\textit{Justification for \eqref{density}:}
Theorem 3.1 of \cite{Leaoetal2004}  and the fact that  $(\real^p, \mathcal{B}(\real^p))$, where $\mathcal{B}(\real^p)$ denotes the Borel sets on $\real^p$, is a Polish space guarantee the existence of the regular conditional probability of $\X\mid\X \in \bs_0 +\spn\{\V\}$  [see also \cite{regProb}]. Further, the measure is concentrated on the  affine subspace $\bs_0 +\spn\{\V\} \subset \real^p$ and is given by \eqref{density} by  Definition 8.38 and Theorem 8.39 of \cite{regProb2} and the orthogonal decomposition \eqref{ortho_decomp}.

\smallskip
\noindent
\textit{Proof of \eqref{LtildeVs0}:}  Since $\X$ and $\epsilon$ in \eqref{mod:basic} are assumed to be independent, 
$\Var(Y\mid\X \in \bs_0 + \spn\{\V\}) = \Var(g(\B^T\X)\mid\X \in \bs_0 + \spn\{\V\}) + \var(\epsilon)$. Using \eqref{density} and $\var(Y \mid Z) = \E(Y^2 \mid Z) - \E(Y \mid Z)^2$, we obtain \eqref{LtildeVs0}.

 We let $\tilde{g}(\V,\bs_0,\rs)= g(\B^T\bs_0 + \B^T\V\rs)^l f_\X(\bs_0 + \V\rs)$. The parameter integral ~\eqref{tl}  is well defined and continuous if (1) $\tilde{g}(\V,\bs_0,\cdot)$ is integrable for all $\V \in \spc(p,q),\bs_0 \in \text{supp}(f_\X)$, (2) $\tilde{g}(\cdot,\cdot,\rs)$ is continuous for all $\rs$, and (3) there exists an integrable dominating function of $\tilde{g}$ that does not depend on  $\V$ and $\bs_0$ [see \cite[p. 101]{HarroHeuser}]. 

Furthermore $t^{(l)}(\V,\bs_0) = \int_{\mathcal{K}} \tilde{g}(\V,\bs_0,\rs) d\rs$ for some compact set $\mathcal{K}$, since $\text{supp}(f_\X)$ is compact due to (A.4). The function $\tilde{g}(\V,\bs_0,\rs)$ is continuous in all inputs by the continuity of $g$ and $f_\X$ by (A.2), and therefore it attains a maximum. In consequence, all three conditions are satisfied so that $t^{(l)}(\V,\bs_0)$ is well defined and continuous.

Next $\mu_l(\V,\bs_0) = t^{(l)}(\V,\bs_0)/t^{(0)}(\V,\bs_0)$ is continuous since $t^{(0)}(\V,\bs_0) > 0$ for all $\bs_0 \in \text{supp}(f_\X)$ by the  continuity of $f_\X$ and 
$\Sigmaxbf > 0$. Then, $\tilde{L}(\V,\bs_0)$ in   \eqref{LtildeVs0} 
is continuous, which results in $L(\V)$ also being well defined and continuous by virtue of it being a parameter integral following the same arguments as above. \qed

\smallskip
\noindent
Next we establish the consistency of the conditional variance estimator. The uniform convergence in probability of the sample objective function in \eqref{LN} is a sufficient condition for obtaining the consistency of $\widehat{\V}_q = \argmin_{\V \in \spc(p,q)} L_n(\V)$, as uniform convergence in probability of a random function implies convergence in probability of the minimizer of $L_n(\V)$ to the minimizer of the limit function. 
Let 
\begin{align}\label{tn}
t^{(l)}_n(\V,\bs_0)=\frac{1}{nh_n^{(p-q)/2}}\sum_{i=1}^n K\left(\frac{d_i(\V,\bs_0)}{h_n}\right)Y^l_i
\end{align}
be the sample version of \eqref{tl} for $l=0,1,2$. The summands of $\tilde{L}_n$ in \eqref{Ltilde} can be expressed  as
\begin{align}\label{yl}
\bar{y}_l(\V,\bs_0)
&= \frac{ t^{(l)}_n(\V,\bs_0)}{t^{(0)}_n(\V,\bs_0)},
\end{align}


Before we start with the proof a few auxiliary lemmas are shown. 

\begin{lemma}\label{aux_lemma2}
Assume  (A.4) and (K.1) hold. Let $Z_n(\V,\bs_0) =  \left(\sum_i g(\X_i)^l K(d_i(\V,\bs_0)/h_n)\right) /(n h_n^{(p-q)/2})$ for a continuous function $g$. Then, 
\begin{align*}
    \E\left(Z_n(\V,\bs_0)\right) 
    &= \int_{\text{supp}(f_\X)\cap\real^{p-q}}K(\|\rs_2\|^2)\int_{\text{supp}(f_\X)\cap\real^q} \tilde{g}(\rs_1,h_n^{1/2}\rs_2)d\rs_1 d\rs_2
\end{align*}
\end{lemma}
where $\tilde{g}(\rs_1,\rs_2) = g(\bs_0 + \V\rs_1 +\U\rs_2)^lf_\X(\bs_0 + \V\rs_1 + \U\rs_2)$, $\xn= \bs_0 + \V \rs_1 + \U \rs_2$  in \eqref{ortho_decomp}.
\begin{proof}[Proof of Lemma~\ref{aux_lemma2}]
By \eqref{ortho_decomp}, $\|\Pbf_{\U} (\xn-\bs_0)\|^2 = \|\U \rs_2\|^2 = \|\rs_2\|^2$. Further 
\begin{gather*} 
\E\left(Z_n(\V,\bs_0)\right) = \frac{1}{h_n^{(p-q)/2}} \int_{\text{supp}(f_\X)} g( \xn)^l K(\|\Pbf_U (\xn-\bs_0)/ h^{1/2}_n \|^2) f_\X(\xn)d\xn \\ 
= \frac{1}{h_n^{(p-q)/2}} \int_{\text{supp}(f_\X)\cap\real^{p-q}}\int_{\text{supp}(f_\X)\cap\real^q} g(\bs_0 + \V \rs_1 +  \U \rs_2)^lK(\|\rs_2/ h^{1/2}_n\|^2) \times \\ f_\X(\bs_0 + \V \rs_1 + \U \rs_2)d\rs_1 d\rs_2\\ \notag
=  \int_{\text{supp}(f_\X)\cap\real^{p-q}}K(\|\rs_2\|^2)\int_{\text{supp}(f_\X)\cap\real^q} g(\bs_0 + \V \rs_1 + h_n^{1/2} \U \rs_2)^l \times \\
f_\X(\bs_0 + \V \rs_1 + h_n^{1/2}\U\rs_2)d\rs_1 d\rs_2
\end{gather*}
where the substitution $\tilde{\rs}_2 = \rs_2/h_n^{1/2}$, $d\rs_2 = h_n^{(p-q)/2} d\tilde{\rs}_2$ was used to obtain the last equality.
\end{proof}
\begin{lemma}\label{aux_lemma3}
Assume (A.1), (A.2), (A.3), (A.4), (H.1) and (K.1) hold. For all $\delta > 0$ there exist an $n^\star$ and finite constants $\tilde{b}^{u,m}$ for $ u \in \{0,1,2,3,4\}$ and $m \in \{1,2\}$ such that  
\begin{equation*}
   \frac{(\tilde{b}^{2l,2} - \delta)}{n h_n^{(p-q)/2}} - \frac{(\tilde{b}^{l,1})^2 +\delta}{n} \leq \var(t^{(l)}_n(\V,\bs_0)) \leq \frac{(\tilde{b}^{2l,2} + \delta)}{n h_n^{(p-q)/2}} - \frac{(\tilde{b}^{l,1})^2 -\delta}{n}
\end{equation*}
\end{lemma}
for $n > n^\star$ and $t^{(l)}_n(\V,\bs_0)$, $l = 0,1,2$, in \eqref{tn}.
\begin{proof}[Proof of Lemma~\ref{aux_lemma3}]
From \eqref{mod:basic} and the binomial formula, $Y_i^l = (g_i + \epsilon_i)^l = \sum_{u=0}^l \binom{l}{u} g_i^{l-u}\epsilon_i^u$ with $g_i = g(\B^T \X_i)$. For $m \in \{1,2\}$ and $l \in \{0,1,...,4\}$, using the independence of $\X_i$ from $\epsilon_i$, we obtain
\begin{align}\label{exp_calc}
    \E\left(Y_i^l K^m(d_i(\V,\bs_0)/h_n)\right) = \sum_{u=0}^l \binom{l}{u} \E\left(g_i^{l-u}K^m(d_i(\V,\bs_0)/h_n)\right)    \E\left(\epsilon_i^u \right). 
\end{align}
Setting $Z_n(\V,\bs_0) = 1/(n h_n^{(p-q)/2}) \sum_i g(\X_i)^{l-u} \tilde{K}(d_i(\V,\bs_0)/h_n)$ in Lemma~\ref{aux_lemma2}, where $\tilde{K}(z) = K^m(z)$ fulfills (K.1) for $m = 1,2$, we obtain  $\E\left(g_i^{l-u}K^m(d_i(\V,\bs_0)/h_n)\right) = h_n^{(p-q)/2}\E(Z_n)$. 
That is, if a kernel satisfies (K.1) its square also satisfies (K.1). Since the integrals are over compact sets by (A.4),  by the dominated convergence theorem and Lemma~\ref{aux_lemma2}, it holds
\begin{align}\label{Zn}
\E(Z_n) =   \int_{\text{supp}(f_\X)\cap\real^{p-q}}\tilde{K}(\|r_2\|^2)\int_{\text{supp}(f_\X)\cap\real^q} \tilde{g}(\rs_1,h_n^{1/2}\rs_2)d\rs_1 d\rs_2 = b_n^{l-u,m} \\
\xrightarrow[n \to \infty]{} b^{l-u,m} = \int_{\text{supp}(f_\X)\cap\real^{p-q}}\tilde{K}(\|\rs_2\|^2)d\rs_2\int_{\text{supp}(f_\X)\cap\real^q} \tilde{g}(\rs_1,0)dr_1 
\end{align}
using that  $\tilde{g}(\rs_1,\rs_2) = g(\B^T\bs_0 + \B^T\V\rs_1 +\B^T\U\rs_2)^{l-u}f_\X(\bs_0 + \V\rs_1 + \U\rs_2)$ is continuous by (A.2) and $h_n \to 0$ by assumption (H.1).

Assumption  (A.3) implies $\E(\epsilon_i^4) < \infty$ for $i=1,\ldots,n$.  From \eqref{Zn} and \eqref{exp_calc}, we obtain
\begin{align} \label{Zmoment}
    \E\left(Y_i^l K^m(d_i(\V,\bs_0)/h_n)/h_n^{(p-q)/2}\right) = \sum_{u=0}^l \binom{l}{u} b_n^{l-u,m}   \E\left(\epsilon_i^u \right) =\tilde{b}_n^{l,m} \\
    \xrightarrow[n \to \infty]{}  \sum_{u=0}^l \binom{l}{u} b^{l-k,m}   \E\left(\epsilon_i^u \right) = \tilde{b}^{l,m} < \infty \notag
\end{align}
By \eqref{Zmoment}, we have for $l=0,1,2$
\begin{align*}
\var\left(t^{(l)}_n(\V,\bs_0)\right) = \frac{1}{n h_n^{p-q}} \var\left(Y_1^l K(d_1(\V,\bs_0)/h_n)\right)
= \frac{\tilde{b}_n^{2l,2}}{n h_n^{(p-q)/2}}  -\frac{(\tilde{b}_n^{l,1})^2}{n}
\end{align*}
since $(Y_i, \X_i^T)_{i=1,...,n}$ are independent draws from the joint distribution of $(Y,\X)$.
This completes the proof since $\tilde{b}_n^{u,m} \to \tilde{b}^{u,m} < \infty$ for $u \in \{0,1,...,4\}$ and $m \in\{1,2\}$.
\end{proof}
Next we show that $d_i(\V,\bs_0)$ in \eqref{distance} is Lipschitz in its inputs under assumption (A.4) in Lemma \ref{d_inequality}. 

\begin{lemma}\label{d_inequality}
Under assumption (A.4) there exists a constant $0 < C_2 < \infty$ such that for all $\delta >0$ and  $\V, \V_j \in \spc(p,q)$ with $\|\Pbf_\V - \Pbf_{\V_j}\| < \delta$ and for all $\bs_0, \bs_j \in \text{supp}(f_\X) \subset \real^p$ with $\|\bs_0 - \bs_j\| < \delta$
\begin{equation*}
    |d_i(\V,\bs_0) - d_i(\V_j,\bs_j)| \leq C_2 \delta
\end{equation*}
for $d_i(\V,\bs_0)$ given by \eqref{distance}
\end{lemma}
\begin{proof}[Proof of Lemma~\ref{d_inequality}]
\begin{gather}
|d_i(\V,\bs_0) - d_i(\V_j,\bs_j)| \leq \left|\|\X_i - \bs_0\|^2 - \|\X_i - \bs_j\|^2\right| + \notag\\
\left|\langle \X_i - \bs_0,\Pbf_{\V}(\X_i - \bs_0)\rangle - \langle \X_i - \bs_j,\Pbf_{\V_j}(\X_i - \bs_j)\rangle\right| = I_1 + I_2\label{di-dj2}
\end{gather}
where $\langle\cdot,\cdot \rangle$ is the scalar product on $\real^p$. For the first term on the right hand side of \eqref{di-dj2}
\begin{align*}
I_1 &= \left|\|\X_i - \bs_0\|^2 - \|\X_i - \bs_j\|^2\right| \leq 
2\left|\langle \X_i,\bs_0 -\bs_j \rangle\right| + \left|\|\bs_0\|^2 - \|\bs_j\|^2\right| \\
& \leq 2\|\X_i\|\|\bs_0 -\bs_j\| + 2C_1\|\bs_0 - \bs_j\| 
\leq 2C_1 \delta + 2C_1\delta = 4 C_1\delta
\end{align*}
by Cauchy-Schwartz and the reverse triangular inequality (i.e. $\left|\|\bs_0\|^2 - \|\bs_j\|^2\right| = \left|\|\bs_0\| - \|\bs_j\|\right|(\|\bs_0\| + \|\bs_j\|) \leq \|\bs_0 - \bs_j\|2C_1$) and $\|\X_i \| \leq \sup_{z \in \text{supp}(f_\X)} \|z \| = C_1 < \infty$ with probability 1 due to (A.4). 
The second term in \eqref{di-dj2} satisfies
\begin{gather*}
I_2 \leq \left|\langle \X_i,(\Pbf_{\V}-\Pbf_{\V_j})\X_i\rangle\right| + 2\left|\langle \X_i,\Pbf_{\V}\bs_0 -\Pbf_{\V_j}\bs_j\rangle\right| + \left|\langle \bs_0,\Pbf_{\V}\bs_0\rangle - \langle \bs_j,\Pbf_{\V_j}\bs_j\rangle\right| \\
\leq \|\X_i\|^2\|\Pbf_{\V} - \Pbf_{\V_j}\| + 2\|\X_i\| \left\|\Pbf_{\V}(\bs_0 - \bs_j) + (\Pbf_{\V}-\Pbf_{\V_j})\bs_j\right\| + \left|\langle \bs_0 - \bs_j,\Pbf_{\V} \bs_0 \rangle\right|+\\
\left|\langle \bs_j,\Pbf_{\V} \bs_0 - \Pbf_{\V_j}\bs_j \rangle\right| \leq C_1^2 \delta + 2C_1(\delta + C_1 \delta) + C_1\delta +C_1 (\delta + C_1 \delta) =4C_1\delta + 4C_1^2 \delta
\end{gather*}
Collecting all constants into $C_2$ (i.e. $C_2 = 8C_1 + 4C_1^2$) yields the result.
\end{proof}
The proofs of Theorems~\ref{thm_L_uniform} and \ref{thm_variance} require  the \textbf{Bernstein inequality} \cite{Bernstein}:
Let $Z_1, Z_2,...$ be an independent sequence of bounded random variables $|Z_i| \leq b$. Let $S_n = \sum_{i=1}^n Z_i$, $E_n = \E(S_n)$ and $V_n = \var(S_n)$. Then,
\begin{equation}\label{Bernstein}
P(|S_n - E_n|>t) < 2 \exp{\left(-\frac{t^2/2}{V_n + b t/3} \right)}
\end{equation}
Furthermore the proof of Theorem \ref{thm_variance} requires assumption (K.2), which obtains
\begin{equation} \label{kernel}
    |K(u) - K(u')| \leq K^*(u') \delta
\end{equation}
for all $u, u'$ with $|u-u'| < \delta \leq L_2$ and $K^*(\cdot)$ is a bounded and integrable kernel function [see \cite{Hansen2008}]. Specifically, if condition (1) of (K.2) holds, then $K^*(u) = L_1 1_{\{|u| \leq 2L_2\}}$. If (2) holds, then  $K^*(u) = L_1 1_{\{|u| \leq 2L_2\}} + 1_{\{|u| > 2L_2\}}|u-L_2|^{-\nu}$.

Let $A = \spc(p,q) \times \text{supp}(f_\X)$ and by a slight abuse of notation, we generically denote constants by $C$. In Theorems \ref{thm_variance} and \ref{thm_bias} we show that the variance and bias terms of \eqref{tn} vanish uniformly in probability, respectively.

\begin{thm}\label{thm_variance}
Under  
(A.1), (A.2), (A.3), (A.4), (K.1), (K.2), $a_n^2 = \log(n)/nh_n^{(p-q)/2} = o(1)$  and $a_n/h_n^{(p-q)/2} = O(1)$, 
\begin{equation}
    \sup_{\V \times \bs_0 \in A} \left|t^{(l)}_n(\V,\bs_0) - \E\left(t_n^{(l)}(\V,\bs_0)\right)\right| = O_{P}(a_n) \quad \text{for} \quad l=0,1,2
\end{equation}
\end{thm}
\begin{remark}
If we assume $|Y| < M_2 < \infty$ almost surely, the requirement $a_n/h_n^{(p-q)/2} = O(1)$ for the bandwidth can be dropped and the truncation step of the proof of Theorem~\ref{thm_variance} can be skipped.
\end{remark}

\begin{proof}[Proof of Theorem~\ref{thm_variance}]
The proof 
is organized in 3 steps: a truncation step, a discretization step by covering $A= \spc(p,q) \times \text{supp}(f_\X)$, and application of Bernstein's inequality \eqref{Bernstein}.

We let $\tau_n = a_n^{-1}$ and  truncate $Y_i^l$ by $\tau_n$ as follows. We let  \begin{align}\label{tl.trc}
    t^{(l)}_{n,\text{trc}}(\V,\bs_0) &= (1/nh_n^{(p-q)/2})\sum_i K(\|\Pbf_\U (\X_i-\bs_0) \|^2/ h_n)Y_i^l 1_{\{|Y_i|^l \leq \tau_n\}}
\end{align} be the truncated version of \eqref{tn} and $\tilde{R}^{(l)}_n = (1/nh_n^{(p-q)/2})\sum_i |Y_i|^l 1_{\{|Y_i|^l > \tau_n\}}$ be the remainder of \eqref{tn}. Therefore $R^{(l)}_n(\V,\bs_0) = t^{(l)}_n(\V,\bs_0) - t^{(l)}_{n,\text{trc}}(\V,\bs_0) \leq M_1 \tilde{R}^{(l)}_n $ due to (K.1)
and
\begin{align}
\sup_{\V \times \bs_0 \in A} \left|t^{(l)}_n(\V,\bs_0) - \E\left(t_n^{(l)}(\V,\bs_0)\right)\right| &\leq M_1(\tilde{R}^{(l)}_n +\E\tilde{R}^{(l)}_n) \notag \\
& \qquad + \sup_{\V \times \bs_0 \in A}\left|t^{(l)}_{n,\text{trc}}(\V,\bs_0) - \E \left( t^{(l)}_{n,\text{trc}}(\V,\bs_0)\right)\right|\label{truncation}
\end{align}
By Cauchy-Schwartz and the Markov inequality, $\Pb(|Z| > t) = \Pb(Z^4 > t^4) \leq \E(Z^4)/t^4$, 
we obtain 
\begin{align}
  \E\tilde{R}^{(l)}_n &= \frac{1}{h_n^{(p-q)/2}}  \E \left(|Y_i|^{l} 1_{\{|Y_i|^l > \tau_n\}}\right) \leq 
  \frac{1}{h_n^{(p-q)/2}}\sqrt{\E(|Y_i|^{2l})} \sqrt{\Pb(|Y_i|^l > \tau_n)} \notag \\
  &\leq \frac{1}{h_n^{(p-q)/2}} \sqrt{\E(|Y_i|^{2l})} \left(\frac{\E(|Y_i|^{4l})}{a_n^{-4}}\right)^{1/2} 
  = o(a_n) \label{Rtilde}
\end{align}
where the last equality uses the assumption $a_n/h_n^{(p-q)/2} = O(1)$ and the expectations are finite due to (A.3) for $l=0,1,2$. Obviously,  no truncation is needed for $l=0$.

Therefore the first two terms of the right hand side of \eqref{truncation} converge to 0 with rate $a_n$ by \eqref{Rtilde} and Markov's inequality. From now to the end of the proof $Y_i$ will denote the truncated version $Y_i 1_{\{|Y_i| \leq \tau_n\}}$ and we do not distinguish the truncated from the untruncated $t_n(\V,\bs_0)$ since this truncation results in an error of magnitude $a_n$. 

For the discretization step we cover the compact set $A = \spc(p,q) \times \text{supp}(f_\X)$ by finitely many balls, which is possible by (A.4) 
 and the compactness of $\spc(p,q)$. 
Let $\delta_n = a_n h_n$ and $A_j = \{\V: \|\Pbf_\V - \Pbf_{\V_j}\| \leq \delta_n\} \times \{\bs :\|\bs - \bs_j\| \leq \delta_n\}$ be a cover of $A$ with ball centers $\V_j \times \bs_j$. Then, $A \subset \bigcup_{j=1}^{N} A_j$ and the number of balls can be bounded by $N \leq C \, \delta_n^{-d}\delta_n^{-p}$ for some constant $C \in (0, \infty)$, where $d = \text{dim}(\spc(p,q)) = pq - q(q+1)/2$. 
Let $\V \times \bs_0 \in A_j$.  Then by Lemma~\ref{d_inequality} there exists  $0 < C_2 < \infty$, such that
\begin{align}\label{inequality1}
|d_i(\V,\bs_0) - d_i(\V_j,\bs_j)| \leq C_2 \delta_n
\end{align}
holds for $d_i$ in \eqref{distance}. Under (K.2), which implies \eqref{kernel},  inequality~\eqref{inequality1} yields
\begin{equation}\label{Ki-Kj}
\left|K\left(\frac{d_i(\V,\bs_0)}{h_n}\right) - K\left(\frac{d_i(\V_j,\bs_j)}{h_n}\right)\right| \leq K^*\left(\frac{d_i(\V_j,\bs_j)}{h_n}\right) C_2 a_n 
\end{equation}
for $\V \times \bs_0 \in A_j$ and $K^*(\cdot)$ an integrable and bounded function. 

Define $r^{(l)}_n(\V_j,\bs_j) = (1/nh_n^{(p-q)/2}) \sum_{i=1}^n K^*(d_i(\V_j,\bs_j)/h_n)|Y_i|^l$. For notational convenience we drop the dependence on $l$ and $j$ in the following and observe that \eqref{Ki-Kj} yields 
\begin{equation}\label{t_diff}
|t^{(l)}_n(\V,\bs_0) - t^{(l)}_n(\V_j,\bs_j)| \leq C_2 a_n r^{(l)}_n(\V_j,\bs_j)
\end{equation}
Since $K^*$ fulfills (K.1) except for continuity, an analogous argument as in the proof of Lemma~\ref{aux_lemma2} yields that $\E\left(r^{(l)}_n(\V_j,\bs_j)\right) < \infty$ by $(A.3)$. By subtracting and adding $t^{(l)}_n(\V_j,\bs_j)$, $\E(t^{(l)}_n(\V_j,\bs_j))$, the triangular inequality, \eqref{t_diff} and integrability of $r_n^l$, we obtain 
\begin{gather}
\left|t^{(l)}_n(\V,\bs_0) - \E\left(t_n^{(l)}(\V,\bs_0)\right)\right| \leq 
\left|t^{(l)}_n(\V,\bs_0) - t^{(l)}_n(\V_j,\bs_j)\right| +  \left|\E\left(t^{(l)}_n(\V_j,\bs_j) - t_n^{(l)}(\V,\bs_0)\right)\right| \notag\\+ \left|t^{(l)}_n(\V_j,\bs_j) - \E\left(t_n^{(l)}(\V_j,\bs_j)\right)\right|  
\leq C_2 a_n \left(|r_n| +   |\E\left(r_n\right)| \right) + \left|t^{(l)}_n(\V_j,\bs_j) - \E\left(t_n^{(l)}(\V_j,\bs_j)\right)\right|  \notag\\
\leq C_2 a_n(|r_n -\E(r_n)| + 2|\E(r_n)|) + \left|t^{(l)}_n(\V_j,\bs_j) - \E\left(t_n^{(l)}(\V_j,\bs_j)\right)\right| \notag\\ 
\leq 2C_3a_n + |r_n -\E(r_n)| + \left|t^{(l)}_n(\V_j,\bs_j) - \E\left(t_n^{(l)}(\V_j,\bs_j)\right)\right| \label{inequality2}
\end{gather}
for any $C_3 > C_2 \E(r^{(l)}_n(\V_j,\bs_j))$ and  $n$ such that $a_n \leq 1$, since $a_n^2 = o(1)$. Summarizing there exists $0 < C_3 < \infty$ such that \eqref{inequality2} holds. 

Then using $\sup_{x \in A} f(x) = \max_{1\leq j\leq N}\sup_{x \in A_j}f(x) \leq \sum_{j=1}^N \sup_{x \in A_j}f(x)$ for any partition of $A$ and continuous function $f$, subadditivty of the probability for the first inequality and \eqref{inequality2} for the third inequality below, it holds
\begin{gather}\label{Prob1}
\Pb(\sup_{\V \times \bs_0 \in A} |t^{(l)}_n(\V,\bs_0) - \E\left(t_n^{(l)}(\V,\bs_0)\right)| > 3C_3a_n) \\ \notag
\leq \sum_{j=1}^N \Pb(\sup_{\V \times \bs_0 \in A_j} |t^{(l)}_n(\V,\bs_0) - \E\left(t_n^{(l)}(\V,\bs_0)\right)| > 3C_3a_n) \\ \notag
\leq N \max_{1 \leq j \leq N} \Pb(\sup_{\V \times \bs_0 \in A_j} |t^{(l)}_n(\V,\bs_0) - \E\left(t_n^{(l)}(\V,\bs_0)\right)| > 3C_3a_n) \\ \notag
\leq N \left(\max_{1 \leq j \leq N}\Pb(|t^{(l)}_n(\V_j,\bs_j) - \E\left(t_n^{(l)}(\V_j,\bs_j)\right)| > C_3 a_n) +  \max_{1 \leq j \leq N} \Pb(|r_n -\E(r_n)| > C_3a_n)\right)  \leq \\ \notag
C\, \delta^{-(d+p)} \left( \max_{1 \leq j \leq N}\Pb(|t^{(l)}_n(\V_j,\bs_j) - \E\left(t_n^{(l)}(\V_j,\bs_j)\right)| > C_3 a_n) + \max_{1 \leq j \leq N} \Pb(|r_n -\E(r_n)| > C_3a_n)\right)
\end{gather}
where the last inequality is due to $N \leq C\, \delta_n^{-d}\delta_n^{-p}$ for a cover of $A$.

Finally, we bound the first and second term in the last line of \eqref{Prob1} by the Bernstein inequality~\eqref{Bernstein}. For the first term in the last line of \eqref{Prob1}, let $Z_i = Y^l_i K(d_i(\V_j,\bs_j)/h_n)$ and $S_n = \sum_i Z_i = nh_n^{(p-q)/2} t^{(l)}_n(\V_j,\bs_j)$, then the $Z_i$ are independent, $|Z_i| \leq b =  M_1\tau_n =M_1/a_n$ by (K.1) and the truncation step. For $V_n = \var(S_n)$, Lemma~\ref{aux_lemma3} yields
\begin{displaymath}
nh_n^{(p-q)/2} \left(\tilde{b}^{2l,2} - \delta - h_n^{(p-q)/2} \left((\tilde{b}^{l,1})^2 + \delta\right) \right) \leq V_n  \leq nh_n^{(p-q)/2} \left(\tilde{b}^{2l,2} + \delta - h_n^{(p-q)/2} \left((\tilde{b}^{l,1})^2  - \delta\right) \right)
\end{displaymath}
for $n$ sufficiently large. We  write $nh_n^{(p-q)/2}C  \geq V_n$ with $C = \tilde{b}^{2l,2} + \delta$
, and set $t = C_3a_n n h_n^{(p-q)/2}$. The  Bernstein inequality~\eqref{Bernstein} yields
\begin{gather*} \label{first_term}
\Pb\left(\left|t^{(l)}_n(\V_j,\bs_j) - \E\left(t_n^{(l)}(\V_j,\bs_j)\right)\right| > C_3 a_n\right)  < 
2 \exp{\left(\frac{-t^2/2}{V_n + b t/3}\right)} \leq \\
2 \exp{\left(-\frac{(1/2)C_3^2a^2_n n^2 h_n^{(p-q)}}{nh_n^{(p-q)/2}C  + (1/3) M_1\tau_n C_3 a_n n h_n^{(p-q)/2})} \right)} \leq 
 2 \exp{\left(-\frac{(1/2)C_3\log(n)}{C/C_3 + (M_1/3) } \right)} = 
2 n^{-\gamma(C_3)}
\end{gather*}
where  $a_n^2 = \log(n)/(n h_n^{(p-q)/2})$ and define $\gamma(C_3) = \frac{(1/2)C_3}{C/C_3 + (M_1 /3) } $, which is  an increasing function that can be made arbitrarily large by increasing  $C_3$.

For the second term in the last line of \eqref{Prob1}, set $Z_i = Y^l_i K^*(d_i(\V_j,\bs_j)/h_n)$ 
in the Bernstein inequality~\eqref{Bernstein} and proceed analogously to obtain
\begin{gather*} \label{second_term}
\Pb\left(\left|r^{(l)}_n(\V_j,\bs_j) - \E\left(r_n^{(l)}(\V_j,\bs_j)\right)\right| > C_3 a_n\right) < 2 n^{- \frac{(1/2)C_3}{C/C_3 + (1/3) M_2}} = 2 n^{-\gamma(C_3)}
\end{gather*}
By (H.1), $h_n^{(p-q)/2} \leq 1$ for $n$ large, so that $\delta_n^{-1} = (a_n h_n)^{-1} \leq n^{1/2}h_n^{-1} h_n^{(p-q)/4} \leq n^{5/2}$.  Further (H.2) implies $1/(nh_n^{(p-q)/2}) \leq 1$ for $n$ large, therefore $h_n^{-1} \leq n^{2/(p-q)} \leq n^2$ since $p-q \geq 1$.  Therefore, \eqref{Prob1} is smaller than $4\,C\, \delta_n^{-(d+p)}n^{-\gamma(C_3)} \leq 4C n^{5(d+p)/2 - \gamma(C_3)}$. For $C_3$ large enough, we have $5(d+p)/2 - \gamma(C_3) < 0$ and $n^{5(d+p)/2 - \gamma(C_3)} \to 0$.  
This completes the proof. 
\end{proof}

\begin{thm}\label{thm_bias}
Under 
(A.1), (A.2) and (A.4), (H.1), (K.1), and 
$\int_{\real^{p-q}}K(\|\rs_2\|^2)d\rs_2 = 1$,
\begin{equation}
    \sup_{\V \times\bs_0 \in A} \left| t^{(l)}(\V,\bs_0)+1_{\{l=2\}}\eta^2t^{(0)}(\V,\bs_0) - \E\left(t_n^{(l)}(\V,\bs_0)\right)\right| =O(h_n), \quad l=0,1,2
\end{equation}
\end{thm}
\begin{proof}[Proof of Theorem~\ref{thm_bias}]
Let $\tilde{g}(\rs_1,\rs_2)= g(\B^T\bs_0 + \B^T\V\rs_1 +\B^T\U\rs_2)^l f_\X(\bs_0 + \V\rs_1 + \U\rs_2)$ where $\rs_1,\rs_2$ satisfy the  orthogonal decomposition~\eqref{ortho_decomp}. Then
\begin{align}\label{bias1}
\E\left(t_n^{(l)}(\V,\bs_0)\right) = \int_{\real^{p-q}}K(\|\rs_2\|^2)\int_{\real^p} \tilde{g}(\rs_1,{h_n}^{1/2}\rs_2) d\rs_1 d\rs_2 + 1_{\{l=2\}} \eta^2 \E\left(t_n^{(0)}(\V,\bs_0)\right)
\end{align}
holds by Lemma~\ref{aux_lemma2} for $l = 0,1$. For $l = 2$, $Y_i^2 = g_i^2 + 2g_i \epsilon_i + \epsilon_i^2$ with $g_i = g(\B^T\X_i)$ and can be handled as in the case of $l = 0,1$. 

 Plugging in \eqref{bias1} the second order Taylor expansion for some $\xi$ in the neighborhood of 0, $\tilde{g}(\rs_1,{h_n}^{1/2}\rs_2) = \tilde{g}(\rs_1,0) + {h_n}^{1/2} \nabla_{\rs_2}\tilde{g}(\rs_1,0)^T\rs_2 + h_n\rs_2^T \nabla^2_{\rs_2}\tilde{g}(\rs_1,\xi) \rs_2$
, yields
\begin{gather*}
\E\left(t_n^{(l)}(\V,\bs_0)\right) = \int_{\real^q}\tilde{g}(\rs_1,0) d\rs_1 + \sqrt{h_n} \left(\int_{\real^q}\nabla_{\rs_2}\tilde{g}(\rs_1,0)d\rs_1\right)^T\int_{\real^{p-q}}K(\|\rs_2\|^2)\rs_2 d\rs_2 + \\
 h_n \frac{1}{2}\int_{\real^{p-q}}K(\|\rs_2\|^2)\int_{\real^p}\rs_2^T \nabla^2_{\rs_2}\tilde{g}(\rs_1,\xi) \rs_2 d\rs_1d\rs_2 =
 t^{(l)}(\V,\bs_0) + h_n \frac{1}{2}R(\V,\bs_0)
\end{gather*}
since $\int_{\real^q}\tilde{g}(\rs_1,0) d\rs_1 = t^{(l)}(\V,\bs_0)$ and $\int_{\real^{p-q}}K(\|\rs_2\|^2)\rs_2 d\rs_2 = 0 \in \real^{p-q}$ due to $K(\|\cdot\|^2)$ being even. Let $R(\V,\bs_0) = \int_{\real^{p-q}}K(\|\rs_2\|^2)\int_{\real^p}\rs_2^T \nabla^2_{\rs_2}\tilde{g}(\rs_1,\xi) \rs_2 d\rs_1d\rs_2$. By (A.4) and (A.2) it holds 
$|\rs_2^T \nabla^2_{\rs_2}\tilde{g}(\rs_1,\xi) \rs_2| \leq C \|\rs_2\|^2$ for $C = \sup_{\xn,\y} \| \nabla^2_{\rs_2}\tilde{g}(\xn,\y)\| < \infty$, since a continuous function over a compact set is bounded. Then, $R(\V,\bs_0) \leq C C_4 \int_{\real^{p-q}}K(\|\rs_2\|^2)\|\rs_2\|^2d\rs_2 < \infty$ for some $C_4 > 0$ since the integral over $\rs_1$ is over a compact set by (A.4). 
\end{proof}

Lemma~\ref{t_uniform} follows directly from Theorems~\ref{thm_variance} and \ref{thm_bias} and the triangle inequality.

\begin{lemma}\label{t_uniform}
Suppose  
(A.1), (A.2), (A.3), (A.4), (K.1), (K.2), (H.1) hold. If $a_n^2 = \log(n)/nh_n^{(p-q)/2} = o(1)$, and $a_n/h_n^{(p-q)/2} = O(1)$, then for $l=0,1,2$
\begin{equation*}
\sup_{\V \times \bs_0 \in A} \left|t^{(l)}(\V,\bs_0)+1_{\{l=2\}}\eta^2t^{(0)}(\V,\bs_0) - t_n^{(l)}(\V,\bs_0)\right| = O_P(a_n + h_n) 
\end{equation*}
\end{lemma}

Combining the results of Theorems~\ref{thm_variance} and \ref{thm_bias} and Lemma \ref{t_uniform} obtains Theorem~\ref{thm_Ltilde_uniform}.

\begin{thm}\label{thm_Ltilde_uniform}
Suppose 
(A.1), (A.2), (A.3), (A.4), (K.1), (K.2), (H.1) hold. Let $a_n^2 = \log(n)/nh_n^{(p-q)/2} = o(1)$,  $a_n/h_n^{(p-q)/2} = O(1)$, $\delta_n = \inf_{\V \times \bs_0 \in A_n}t^{(0)}(\V,\bs_0)$, where $t^{(0)}(\V,\bs_0)$ is defined in \eqref{tl}, andand $A_n = \spc(p,q) \times \{\xn \in \text{supp}(f_\X): |\xn - \partial\text{supp}(f_\X)| \geq b_n\}$, where $\partial C $ denotes the boundary of the set $C$ and $|\xn - C| = \inf_{\rs \in C} |\xn - \rs| $, for a sequence $b_n \to 0$ so that   $\delta_n^{-1}(a_n + h_n) \to 0$ for any bandwidth $h_n$ that satisfies the assumptions. Then,
\begin{equation*}
\sup_{\V \times \bs_0 \in A}\left|\bar{y}_l(\V,\bs_0) - \mu_l(\V,\bs_0)-1_{\{l=2\}}\eta^2t^{(0)}(\V,\bs_0)\right| = O_P(\delta_n^{-1}(a_n + h_n)), \quad l=0,1,2
\end{equation*}
and
\begin{equation}\label{Ltilde_uniform}
\sup_{\V \times \bs_0 \in A}\left|\tilde{L}_n(\V,\bs_0) - \tilde{L}(\V,\bs_0)\right| = O_P(\delta_n^{-1}(a_n + h_n))
\end{equation}
where $\bar{y}_l(\V,\bs_0)$, $\mu_l(\V,\bs_0)$, $\tilde{L}_n(\V,\bs_0)$ and $\tilde{L}(\V,\bs_0)$ are defined in \eqref{yl}, \eqref{mu_l},  \eqref{Ltilde} and \eqref{LtildeVs0}, respectively. 
\end{thm}
\smallskip
\begin{proof}[Proof of Theorem~\ref{thm_Ltilde_uniform}]
\begin{equation*}
 \bar{y}_l(\V,\bs_0) = \frac{t_n^{(l)}(\V,\bs_0)}{t_n^{(0)}(\V,\bs_0)} =   \frac{t_n^{(l)}(\V,\bs_0)/t^{(0)}(\V,\bs_0)}{t_n^{(0)}(\V,\bs_0)/t^{(0)}(\V,\bs_0)}
\end{equation*}
We consider the numerator and enumerator separately. 
By Lemma~\ref{t_uniform} 
\begin{gather*}
\sup_{\V \times \bs_0 \in A_n} \left|\frac{t_n^{(0)}(\V,\bs_0)}{t^{(0)}(\V,\bs_0)} - 1\right| 
\leq \frac{\sup_{A}|t_n^{(0)}(\V,\bs_0) - t^{(0)}(\V,\bs_0)|}{\inf_{A_n} t^{(0)}(\V,\bs_0)} = O_P(\delta_n^{-1}(a_n + h_n))
\end{gather*}
 Next 
\begin{gather*}
\sup_{\V \times \bs_0 \in A_n} \left|\frac{t_n^{(l)}(\V,\bs_0)}{t^{(0)}(\V,\bs_0)} - \mu_l(\V,\bs_0)\right| 
\leq \frac{\sup_{A}|t_n^{(l)}(\V,\bs_0) - t^{(l)}(\V,\bs_0)|}{\inf_{A_n} t^{(0)}(\V,\bs_0)} = O_P(\delta_n^{-1}(a_n + h_n)).
\end{gather*}
Therefore by $A_n \uparrow A = \spc(p,q) \times \text{supp}(f_\X)$ we get
\begin{equation*}
    \lim_{n \to \infty} \sup_{\V \times \bs_0 \in A_n}\left|\frac{t_n^{(l)}(\V,\bs_0)}{t^{(0)}(\V,\bs_0)} - \mu_l(\V,\bs_0)\right| = \lim_{n \to \infty} \sup_{\V \times \bs_0 \in A}\left|\frac{t_n^{(l)}(\V,\bs_0)}{t^{(0)}(\V,\bs_0)} - \mu_l(\V,\bs_0)\right| 
\end{equation*}
and in total we obtain
\begin{equation*}
\bar{y}_l(\V,\bs_0) = \frac{t_n^{(l)}(\V,\bs_0)/t^{(0)}(\V,\bs_0)}{t_n^{(0)}(\V,\bs_0)/t^{(0)}(\V,\bs_0)} = \frac{\mu_l + O_P(\delta_n^{-1}(a_n + h_n))}{1 + O_P(\delta_n^{-1}(a_n + h_n))} = \mu_l + O_P(\delta_n^{-1}(a_n + h_n)).
\end{equation*}
For $l = 2$,  $Y^2_i = g(\B^T\X_i)^2 + 2g(\B^T\X_i)\epsilon_i + \epsilon_i^2$, and \eqref{Ltilde_uniform} follows from~\eqref{LtildeVs0}.
\end{proof}

\begin{lemma}\label{mu_lemma}
Under  (A.1), (A.2), (A.4), there exists $0 < C_5 < \infty$ such that
\begin{align}\label{mu_inequality}
    \left|\mu_l(\V,\bs_0) - \mu_l(\V_j,\bs_0)\right| \leq C_5 \|\Pbf_\V -\Pbf_{\V_j}\|
\end{align}
for all $\bs_0 \in \text{supp}(f_\X)$
\end{lemma}
\begin{proof}
From  the representation $\tilde{t}^{(l)}(\Pbf_\V,\bs_0)$  in \eqref{Grassman} instead of $t^{(l)}(\V,\bs_0)$, we consider $\mu_l(\V,\bs_0) = \mu_l(\Pbf_\V,\bs_0)$ as a function on the Grassmann manifold. 
Then,
\begin{align}\label{LipschitG}
    \left|\mu_l(\Pbf_\V,\bs_0) - \mu_l(\Pbf_{\V_j},\bs_0)\right| &= \left|\frac{\tilde{t}^{(l)}(\Pbf_\V,\bs_0)}{\tilde{t}^{(0)}(\Pbf_\V,\bs_0)} - \frac{\tilde{t}^{(l)}(\Pbf_{\V_j},\bs_0)}{\tilde{t}^{(0)}(\Pbf_{\V_j},\bs_0)}\right| \notag \\ 
   &\leq \frac{\sup |\tilde{t}^{(0)}(\Pbf_\V,\bs_0)|}{(\inf \tilde{t}^{(0)}(\Pbf_\V,\bs_0))^2}\left| \tilde{t}^{(l)}(\Pbf_\V,\bs_0)-\tilde{t}^{(l)}(\Pbf_{\V_j},\bs_0)\right|\notag \\
    &\quad +\frac{\sup \tilde{t}^{(l)}(\Pbf_\V,\bs_0)}{(\inf \tilde{t}^{(0)}(\Pbf_\V,\bs_0))^2}\left| \tilde{t}^{(0)}(\Pbf_\V,\bs_0)-\tilde{t}^{(0)}(\Pbf_{\V_j},\bs_0)\right| 
\end{align}
with $\sup_{\Pbf_\V \in Gr(p,q)} \tilde{t}^{(0)}(\Pbf_\V,\bs_0) \in (0,\infty)$ and  $\inf_{\Pbf_\V \in Gr(p,q)} \tilde{t}^{(0)}(\Pbf_\V,\bs_0) \in (0,\infty)$ since $\tilde{t}^{(l)}$ is continuous, $\Sigmaxbf >0$ and $\bs_0 \in \text{supp}(f_\X)$.

By (A.2), $\tilde{g}(\xn) =g(\B^T \xn)f_\X(\xn)$ is twice continuous differentiable  and therefore Lipschitz continuous on compact sets.  We denote its Lipschitz constant by $L < \infty$. Therefore,
\begin{gather} 
    \left| \tilde{t}^{(l)}(\Pbf_{\V},\bs_0)-\tilde{t}^{(l)}(\Pbf_{\V_j},\bs_0)\right| \leq \int_{\text{supp}(f_\X)} \left|\tilde{g}(\bs_0 + \Pbf_{\V} \rs)-\tilde{g}(\bs_0 + \Pbf_{\V_j} \rs)\right|d \rs  \notag \\ \leq
    L \int_{\text{supp}(f_\X)}  \|(\Pbf_{\V} -\Pbf_{\V_j}) \rs\|d\rs \leq
    L\left(\int_{\text{supp}(f_\X)}  \| \rs \|dr\right) \|\Pbf_\V -\Pbf_{\V_j}\| \label{t_inequality}
\end{gather}
where the last inequality is due to the sub-multiplicativity of the Frobenius norm and the integral being finite by (A.4). Plugging \eqref{t_inequality} in \eqref{LipschitG} and collecting all constants into $C_5$ yields \eqref{mu_inequality}. 
\end{proof}

\begin{proof}[Proof of Theorem~\ref{thm_L_uniform}]
By \eqref{LN} and \eqref{objective},
\begin{align}
\left|L_n(\V) - L(\V)\right| \leq \left|\frac{1}{n} \sum_i \left(\tilde{L}_n(\V,\X_i) -\tilde{L}(\V,\X_i)\right)\right| + \left|\frac{1}{n} \sum_i \left(\tilde{L}(\V,\X_i) - \E(\tilde{L}(\V,\X))\right) \right|   \label{Ln-L}
\end{align}
The first term on the right hand side of \eqref{Ln-L} goes to 0 in probability uniformly in $\V$ by Theorem~\ref{thm_Ltilde_uniform}, 
\begin{equation}
\left|\frac{1}{n} \sum_i \tilde{L}_n(\V,\X_i) -\tilde{L}(\V,\X_i)\right| \leq \sup_{\V \times \bs_0 \in A}\left|\tilde{L}_n(\V,\bs_0) - \tilde{L}(\V,\bs_0)\right| = O_P(\delta_n^{-1} (a_n + h_n))
\end{equation}
The second term in \eqref{Ln-L} converges to 0 almost surely for all $\V \in \spc(p,q)$ by the strong law of large numbers. In order  to show uniform convergence the same technique as in the proof of Theorem~\ref{thm_variance} is used. Let $B_j = \{\V \in \spc(p,q): \|\V\V^T - \V_j\V_j^T\| \leq \tilde{a}_n\}$ be a cover of $\spc(p,q)\subset \bigcup_{j=1}^{N} B_j$ with $N \leq C\, \tilde{a}_n^{-d} = C\,(n/\log(n))^{d/2} \leq C\, n^{d/2}$, where $d = \dim(\spc(p,q))$ is defined in the proof of Theorem~\ref{thm_variance}. By Lemma~\ref{mu_lemma}, 
\begin{align}\label{inequality3}
    \left|\mu_l(\V,\X_i) - \mu_l(\V_j,\X_i)\right| \leq C_5 \|\Pbf_\V  - \Pbf_{\V_j}\|
\end{align}
Let $G_n(\V) = \sum_i\tilde{L}(\V,\X_i)/n$ with $\E(G_n(V)) = L(\V)$. Using \eqref{inequality3} and following the same steps as in the proof of Theorem~\ref{thm_variance} we obtain
\begin{align}\label{G_n_ineq}
    \left|G_n(\V)-L(\V)\right| &\leq \left|G_n(\V) - G_n(\V_j)\right| + \left|G_n(\V_j) -L(\V_j)\right| + \left|L(\V)-L(\V_j)\right| \notag \\ 
    & \leq 2C_6\tilde{a}_n + \left|G_n(\V_j) -L(\V_j)\right| 
\end{align}
for $\V \in B_j$ and some $C_6 > C_5$. Inequality \eqref{G_n_ineq} leads to
\begin{align} 
    \Pb\left(\sup_{\V \in \spc(p,q)}|G_n(\V) - L(\V)| > 3C_6\tilde{a}_n\right) \leq C\,N\, \Pb(\sup_{\V \in B_j}|G_n(\V) - L(\V)| > 3C_6\tilde{a}_n) \notag \\ 
    \leq C\, n^{d/2} \Pb(|G_n(\V_j) -L(\V_j)|> C_6\tilde{a}_n) \leq C\, n^{d/2} n^{-\gamma(C_6)} \to 0 \label{inequality5}
\end{align}
where the last inequality in \eqref{inequality5} is due to the Bernstein inequality \eqref{Bernstein} with $Z_i = \tilde{L}(\V_j,\X_i)$, which is bounded since  $\tilde{L}(\cdot,\cdot)$ is continuous on the compact set $A$, and $\gamma(C_6)$ a monotone increasing function of $C_6$ that can be made arbitrarily large by choosing $C_6$ accordingly. Therefore, $\sup_{\V \in \spc(p,q)}\left|L_n(\V) - L(\V)\right| \leq O_P(\delta_n^{-1}(a_n + h_n) +\tilde{a}_n)$ with $\delta_n = \inf_{\V \times \bs_0 \in A_n}t^{(0)}(\V,\bs_0)$, where $t^{(0)}(\V,\bs_0)$ is defined in \eqref{tl}, and $A_n = \spc(p,q) \times \{\xn \in \text{supp}(f_\X): f_\X(\xn) \geq b_n\}$ for a sequence $b_n \to 0$ so that   $\delta_n^{-1}(a_n + h_n) \to 0$ for any bandwidth $h_n$ that satisfies the assumptions, which implies  \eqref{thm5_eq}.
\end{proof}
\begin{proof}[Proof of Theorem~\ref{thm_consistency}]
We apply Theorem 4.1.1 of \cite{Takeshi} to obtain consistency of the conditional variance estimator. This theorem requires three conditions that guarantee the convergence of the minimizer of a sequence of random functions $L_n(\Pbf_\V)$ to the minimizer of the limiting function $L(\Pbf_\V)$; i.e., $\Pbf_{\spn\{\widehat{\B}\}^\perp} = \argmin L_n(\Pbf_\V) \to \Pbf_{\spn\{\B\}^\perp} = \argmin L(\Pbf_\V)$ in probability.
To apply the theorem three conditions have to be met:
(1) The parameter space 
is compact; 
(2) $L_n(\V)$ is continuous and a measurable function of the data $(Y_i,\X_i^T)_{i=1,...,n}$ and 
    (3) $L_n(\V)$ converges uniformly to $L(\V)$ and $L(\V)$ attains a unique global minimum at $\spn\{\B\}^\perp$. 

Since $L_n(\V)$ depends on $\V$ only through $\Pbf_\V = \V\V^T$,  $L_n(\V)$ can be considered as functions on the Grassmann manifold, which is compact, and the same holds true for $L(\V)$ by \eqref{Grassman}. 
Further,  $L_n(\V)$ is by definition a measurable function of the data and continuous in $\V$ if  a continuous kernel is used, such as the  Gaussian. Theorem~\ref{thm_L_uniform} obtains  the uniform convergence and Theorem~\ref{thm1} that the minimizer is  unique when $L(\V)$ is minimized  over the Grassmann manifold $G(p,q)$, since $\spn\{\B\}$ is uniquely identifiable and so is   $\spn\{\B\}^\perp$ (i.e. $\|\Pbf_{\spn\{\widehat{\B}\}} - \Pbf_{\spn\{\B\}}\|=\|\widehat{\B}\widehat{\B}^T - \B\B^T\| = \| (\I_p- \B\B^T)- (\I_p-\widehat{\B}\widehat{\B}^T)\| = \|\Pbf_{\spn\{\widehat{\B}\}^\perp} - \Pbf_{\spn\{\B\}^\perp}\|$). Thus, all three conditions are met and the result is obtained.
\end{proof}

\begin{proof}[Proof of Theorem~\ref{lemma-one}]
The Gaussian kernel $K$ satisfies $\partial_z K(z) = - z K(z)$. From \eqref{weights} and \eqref{Ltilde} we have $\tilde{L}_n = \bar{y}_2 - \bar{y}_1^2$ where $\bar{y}_l = \sum_i w_i Y_i^l$, $l=1,2$. We let $K_{j} = K(d_j(\V,\bs_0)/h_n)$,
suppress the dependence on $\V$ and $\bs_0$ and write $w_i = K_i/\sum_j K_j$. Then, $\nabla K_i = (-1/h_n^2)K_i d_i \nabla d_i$ and $\nabla w_i =  -\left(K_i d_i \nabla d_i (\sum_j K_j) - K_i \sum_j K_j d_j\nabla d_j\right)/(h_n \sum_j K_j)^2$. Next,
\begin{align}
\nabla \bar{y}_l &= -\frac{1}{h_n^2}\sum_i Y_i^l\frac{\left(K_id_i\nabla d_i - K_i (\sum_jK_jd_j\nabla d_j)\right)}{( \sum_jK_j)^2} 
= -\frac{1}{h_n^2}\sum_i Y_i^l w_i \left(d_i \nabla d_i - \sum_j w_j d_j \nabla d_j\right) \notag \\
&= -\frac{1}{h_n^2}\left(\sum_i Y_i^l w_i d_i \nabla d_i - \sum_jY_j^l w_j \sum_i w_id_i \nabla d_i\right) 
=-\frac{1}{h_n^2}\sum_i (Y_i^l - \bar{y}_l) w_i d_i \nabla d_i \label{grad.yl}
\end{align}
Then, $\nabla \tilde{L}_n = \nabla \bar{y}_2 - 2\bar{y}_1 \nabla \bar{y}_1$, and inserting $\nabla \bar{y}_l$ from \eqref{grad.yl} yields $\nabla \tilde{L}_n = (-1/h_n^2)\sum_i (Y_i^2 -\bar{y}_2 - 2\bar{y}_1(Y_i - \bar{y}_1))w_i d_i \nabla d_i = (1/h_n^2)(\sum_i \left(\tilde{L}_n -(Y_i -\bar{y}_1)^2 \right) w_i d_i \nabla d_i)$, since $Y_i^2 -\bar{y}_2 - 2\bar{y}_1(Y_i - \bar{y}_1) = (Y_i -\bar{y}_1)^2 - \tilde{L}_n $.
\end{proof}
\end{document}